\documentclass[11pt,reqno,a4paper,nofootinbib]{article}
\usepackage{amsmath}
\usepackage{amsthm}
\usepackage{latexsym}
\usepackage{amsfonts}
\usepackage{amssymb}
\usepackage{color}
\usepackage{bbm}
\usepackage{dsfont}
\usepackage{graphicx}
\usepackage{textcomp}

\usepackage{authblk}

\newtheorem{theorem}{Theorem}
\newtheorem{proposition}[theorem]{Proposition}
\newtheorem{lemma}[theorem]{Lemma}
\newtheorem{corollary}[theorem]{Corollary}

\theoremstyle{definition}

\newtheorem{remark}[theorem]{Remark}
\newtheorem{example}[theorem]{Example}

\newcommand{\be}{\begin{equation}}
\newcommand{\ee}{\end{equation}}
\newcommand{\bea}{\begin{eqnarray}}
\newcommand{\eea}{\end{eqnarray}}

\newcommand{\binH}{H}
\newcommand{\binrel}{D_2}

\newcommand{\var}{{\rm var}} 
\newcommand{\complex}{{\mathbb C}} 
\newcommand{\nat}{{\mathbb N}} 

\newcommand{\ket}[1]{|#1\rangle} 
\newcommand{\bra}[1]{\langle#1|} 
\newcommand{\tr}[1]{{\rm{tr}}\left[#1\right]} 
\newcommand{\1}{\mathbbm{1}}

\newcommand{\R}{\mathbb{R}}
\newcommand{\cB}{{\mathcal{B}}}

\newcommand{\supp}{{\rm supp}}

\newcommand{\ii}{\mathbbm{1}}

\newcommand{\C}{\mathbb{C}}

\def\>{{\rangle}}
\def\<{{\langle}}

\def\lsim{\mathrel{\rlap{\lower4pt\hbox{\hskip1pt$\sim$}}
    \raise1pt\hbox{$<$}}}                
\def\gsim{\mathrel{\rlap{\lower4pt\hbox{\hskip1pt$\sim$}}
    \raise1pt\hbox{$>$}}}                

\newcommand{\rank}[1]{\mathrm{rank}(#1)} 

\title{\vspace{-1.2cm}{\vspace{-0.5cm}\Large\textbf{An improved Landauer Principle with finite-size corrections}}}
\author{\vspace{-0.2cm}David Reeb\thanks{david.reeb@tum.de, m.wolf@tum.de}~}
\author{Michael M.~Wolf$^*$}
\affil{\vspace{-0.25cm}\small{Department of Mathematics, Technische Universit\"at M\"unchen, 85748 Garching, Germany}}
\date{}

\begin{document}

\maketitle

\vspace{-1cm}\begin{abstract}Landauer's Principle relates entropy decrease and heat dissipation during logically irreversible processes. Most theoretical justifications of Landauer's Principle either use thermodynamic reasoning or rely on specific models based on arguable assumptions. Here, we aim at a general and minimal setup to formulate Landauer's Principle in precise terms. We provide a simple and rigorous proof of an improved version of the Principle, which is formulated in terms of an equality rather than an inequality. The proof is based on quantum statistical mechanics concepts rather than on thermodynamic argumentation. From this equality version, we obtain explicit improvements of Landauer's bound that depend on the effective size of the thermal reservoir and reduce to Landauer's bound only for infinite-sized reservoirs.\end{abstract}

\vspace{-0.4cm}
\tableofcontents\vspace{-0.13cm}

\newpage

\section{Introduction}\label{introductionsection}
The Maxwell's Demon paradox suggested that one can lower the entropy of a gas of particles without expending energy, and thus violate the Second Law of Thermodynamics, if one has information about the positions and momenta of the particles \cite{maxwell1871,maxwelldemonbook}. During the resolution of this puzzle it became however clear that thermodynamics imposes physical constraints on information processing \cite{szilard29,vNeumannMathemGrundlagen}. Rolf Landauer \cite{landauer} recognized that it is the logically irreversible erasure of information that necessitates a corresponding entropy increase in the environment \cite{bennettreview82}; i.e.\ information erasure from the information-bearing degrees of freedom of a memory register or computer causes entropy to flow to the non-information-bearing degrees of freedom. At inverse temperature $\beta$, this entropy increase causes heat
\begin{align}
\Delta Q~\geq~\Delta S/\beta\label{landauersbasicinequality}
\end{align}
to be dissipated, where $\Delta S$ denotes the entropy decrease in the memory. This consequence is \emph{Landauer's Principle}, and the inequality (\ref{landauersbasicinequality}) is also called the \emph{Landauer bound} or \emph{limit}.

Since its inception \cite{landauer}, the above argument has been controversially discussed on different levels. For example, it has been disputed whether it is necessary to assume the validity of the Second Law of Thermodynamics in order to derive Landauer's Principle or whether, conversely, the Second Law itself is actually a consequence of the Principle (see e.g.\ \cite{earmannorton,bennettnotes03,maxwelldemonbook}). Situations have been reported -- both theoretically \cite{breakdownofLP,alickicorrelations} and in experiment \cite{expmuchmore} -- which supposedly violate Landauer's Principle. And it was actually already recognized by Bennett \cite{bennettreversible73,bennettreview82} that all computation \emph{can} be done reversibly, thereby avoiding irreversible erasure and requiring no heat dissipation in principle. On the other hand, the Principle was successful in exorcising Maxwell's Demon \cite{bennettreview82,maxwelldemonbook}, and a recent experiment \cite{experimentnature} approached Landauer's limit but could not surpass it. Attempts to formulate and prove Landauer's Principle by more microscopic methods followed later (e.g.\ \cite{piechocinska,sagawaueda09}), but they still have deficiencies as we discuss more detail in Section \ref{subsectpreworks}.

\medskip

Much of the misunderstanding and controversy around Landauer's Principle appears to be due to the fact that its general statement has not been written down formally or proved in a rigorous way in the framework of quantum statistical mechanics. To remedy this situation is the first goal of the present work.

We formulate in precise mathematical and statistical mechanics terms the setup for Landauer's Principle. The four assumptions are listed at the beginning of Section \ref{setupsubsection} (see also Fig.\ \ref{figsetup} for an overview of the setup). Our formulation encompasses processes more general than ``erasure'', and the setting is minimal in the sense that Landauer's bound can be violated when any one of our assumptions is dropped.

Our first main result is a proof of Landauer's Principle in the form of a sharpened equality version (Theorem \ref{landauereqntheorem}):\begin{equation}\label{statelandauerequalityinintro}\beta\Delta Q~=~\Delta S\,+\,I(S':R')\,+\,D(\rho'_R\|\rho_R)~\geq~\Delta S~,\end{equation}where the mutual information $I(S':R')$ quantifies the correlations built up between system and reservoir during the process and the relative entropy $D(\rho'_R\|\rho_R)$ can be physically interpreted as the free energy increase in the reservoir. Closer examination reveals that Landauer's bound $\beta\Delta Q\geq\Delta S$ can be tight only if $\Delta S=0$. The Landauer bound (\ref{landauersbasicinequality}) can thus be improved for all non-trivial processes.

Our second main result is then an explicit improvement of Landauer's bound (Section \ref{finitesizesect}), which will be possible when the thermal reservoir assisting in the process has a finite Hilbert space dimension $d<\infty$. A paradigmatic result is here (see Theorem \ref{maintheoremcombined}):\begin{align}\label{paradigmaticfinitiesizeeffect}\beta\Delta Q~\geq~\Delta S+\frac{2(\Delta S)^2}{\log^2(d-1)+4}\qquad\text{for~any~physical~process~in~which}~\Delta S\,\geq\,0~.\end{align}This is illustrated in Fig.\ \ref{relentropygraphs}:\ For small reservoirs, the necessary heat expenditure $\Delta Q$ lies several ten percent above the Landauer limit (\ref{landauersbasicinequality}). The main technical tool in deriving these finite-size effects is a tight entropy inequality between relative entropy and entropy difference \cite{inequalityarxivpaper}.

In Section \ref{extendednotionssection} and Appendix \ref{extendednotionssectionapp}, we present a few extensions of the setup from Section \ref{setupsubsection}. Section \ref{attainingsection} forms a counterpart to results like Eq.\ (\ref{paradigmaticfinitiesizeeffect}), as we construct processes that approach Landauer's bound (\ref{landauersbasicinequality}) arbitrarily closely by using a reservoir of unbounded size $d$.

\section{Setup and notation}\label{setupsect}
Here we formalize the exact setting in which we prove -- and improve -- Landauer's Principle. We avoid unnecessary excess structure that is present in some previous works (discussed in Section \ref{subsectpreworks}), and aim to motivate each necessary ingredient. This is the first step to a rigorous treatment of Landauer's Principle in Sections \ref{Lprinciplesharpened} and \ref{finitesizesect}. Our setup and the subsequent statements will be quantum-mechanical, but apply to the classical (probabilistic) case as well upon restriction to commuting states and Hamiltonians. In Section \ref{extendednotionssection} and Appendix \ref{extendednotionssectionapp} we discuss some extensions of the setup described here.

\subsection{Setup of the process}\label{setupsubsection}As commonly conceived, Landauer's process \cite{landauer} is supposed to ``erase'' or ``reset'' the state of a system by having it ``interact'' with a ``thermal reservoir'', bringing the system into a ``definite'' state, such as a fixed pure state. We use this conception as a motivation, but our setup will be more general and precise.

\bigskip

\noindent The four assumptions needed for Landauer's Principle are as follows (see also Fig.\ \ref{figsetup}):
\begin{enumerate}
\item[(a)]\vspace{-0.1cm}the process involves a ``system'' $S$ and a ``reservoir'' $R$, both described by Hilbert spaces,
\item[(b)]\vspace{-0.15cm}the reservoir $R$ is initially in a thermal state, $\rho_R=e^{-\beta H}/\tr{e^{-\beta H}}$, where $H$ is a Hermitian operator on $R$ (``Hamiltonian'') and $\beta\in[-\infty,+\infty]$ is the ``inverse temperature'',
\item[(c)]\vspace{-0.15cm}the system $S$ and the reservoir $R$ are initially uncorrelated, $\rho_{SR}=\rho_S\otimes\rho_R$,
\item[(d)]\vspace{-0.15cm}the process itself proceeds by unitary evolution, $\rho'_{SR}=U\rho_{SR}U^{\dagger}$.
\end{enumerate}
\vspace{-0.05cm}We now discuss each of these four assumptions in more detail, arguing that this setup is minimal.

\bigskip

The process acts on two subsystems, $S$ and $R$, and we call $S$ the ``system'' and $R$ the ``reservoir''. We model these as quantum systems with Hilbert spaces of finite dimensions $d_S$ and $d\equiv d_R$, respectively (see Fig.~\ref{figsetup}). The extension of our treatment to infinite-dimensional state spaces is discussed in Appendix \ref{sectioninfinitedim}.

Secondly, we require a Hamiltonian $H_R\equiv H$ of the reservoir to be given, i.e.\ a Hermitian operator $H=H^\dagger\in\cB(\complex^d)$. We furthermore assume that initially, i.e.\ before the process starts, the reservoir is in a thermal state
\begin{equation}\label{definethermalstateinsetup}
\rho_R~=~\frac{e^{-\beta H}}{\tr{e^{-\beta H}}}
\end{equation}
at inverse temperature $\beta\in[-\infty,+\infty]$ (for $\beta=\pm\infty$, $\rho_R$ is the maximally mixed state on the ground space of $\pm H$; see Appendix \ref{thermodynamicappendix} for more details on statistical mechanics).

The assumption of an initially thermal reservoir state to be consumed during the process is reasonable, since such states may be considered cheaply available; this is due to the physics intuition -- verified in several computable models -- that a reservoir governed by $H$ and at inverse temperature $\beta$ will generically relax to the state (\ref{definethermalstateinsetup}) (see e.g.\ \cite{brattelirobinson,thirringstatphys}) or to a state close to it (e.g.\ \cite{muellerwiebe}). Furthermore, if $\rho_R$ would not be a thermal state as in Eq.\ (\ref{definethermalstateinsetup}), then it would be possible to violate Landauer's bound (see Section \ref{subsectpreworks}); this is related to the fact that thermal states are the only \emph{completely passive} states, meaning that from an arbitrary number of state copies one cannot extract work by unitary operations alone \cite{puszworonowicz}. Lastly, the appearance of the initial thermal reservoir state (\ref{definethermalstateinsetup}) is motivated already by the mathematical need to have a definite value for $\beta$ appearing in Landauer's bound as stated in Eq.\ (\ref{landauersbasicinequality}).

We do not put any assumption on the initial system state $\rho_S$; in particular, it need not be a thermal state. The following developments and results are in fact completely independent of the system Hamiltonian, and it does not even need to be specified.

As the third assumption, we require system and reservoir to be initially uncorrelated, i.e.
\begin{align}\label{productstateassumptioninsetup}
\rho_{SR}~=~\rho_S\otimes\rho_R~.
\end{align}
This assumption will be important for Landauer's Principle to hold: If the initial state $\rho_{SR}$ would for instance be such that the reservoir $R$ had perfect classical correlations with $S$, then a unitary process could reduce the system entropy without any heat dissipation, in violation of Landauer's bound. This can be seen from the example in Section \ref{memorysection}, in which the system $S$ is brought into a final pure state without any change of the reduced state of the reservoir.

The product state assumption (\ref{productstateassumptioninsetup}) is standard in the theory of thermodynamics and in many common tasks in information processing, e.g.\ in the paradigmatic examples of resetting a register in a computer or when performing error correction. In these cases, the assisting reservoir $R$ is often assumed not to have previously interacted with the register $S$, such that their states are independent. When however system and reservoir have undergone prior interactions, they may be correlated, and we give a Landauer-like bound for this case in Eq.\ (\ref{verygeneralsecondlawwithmemory}). (The case of extreme correlations would correspond to reversible computation \cite{bennettreversible73,bennettreview82}, which indeed does not require any energy expenditure but suffers from error build-up in practical implementations.) Furthermore, as we are assuming an initially thermal reservoir (see Eq.\ (\ref{definethermalstateinsetup})), correlations in $SR$ would be unnatural unless the full initial state $\rho_{SR}$ were also thermal, but this would then require a (Hamiltonian) interaction term between $S$ and $R$; see also \cite{relativethermalization}. Some reported ``violations'' of Landauer's bound \cite{breakdownofLP,expmuchmore} can be explained by their not respecting the initial product state assumption (\ref{productstateassumptioninsetup}).

Extensions of the product state assumption (\ref{productstateassumptioninsetup}) and the inclusion of an additional memory register are discussed in Section \ref{memorysection} (see also Section \ref{correlationsubsection}).

\begin{figure}[t!]
\centering
\includegraphics[scale=0.9]{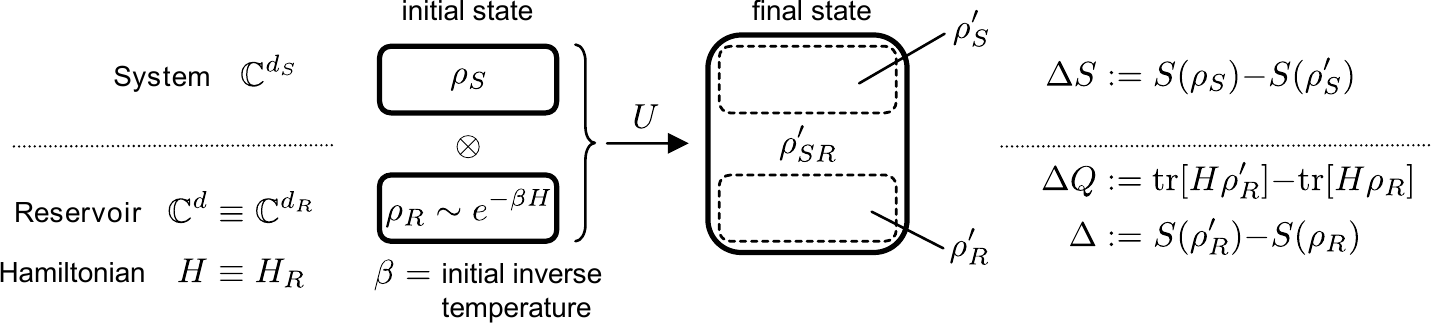}
\caption{\label{figsetup}Our formalization of the ``process'':~An initially uncorrelated state $\rho_S\otimes\rho_R$ of a system $S$ and a thermal reservoir $R$ ($\rho_R=e^{-\beta H}/\tr{e^{-\beta H}}$) is subjected to a unitary evolution, $\rho'_{SR}:=U(\rho_S\otimes\rho_R)U^\dagger$. This replaces $\rho_S$ by the marginal state $\rho'_S={\rm tr}_R[\rho'_{SR}]$. \emph{Landauer's Principle} (or \emph{Landauer's bound}) \cite{landauer} now relates the entropy decrease $\Delta S$ in the system to the heat $\Delta Q$ dissipated to the reservoir:~$\beta\Delta Q\geq\Delta S$. We rigorously prove and improve this inequality in the paper.}
\end{figure}

\bigskip

So far we have described the assumptions on the initial state $\rho_{SR}$, Eqs.\ (\ref{definethermalstateinsetup}) and (\ref{productstateassumptioninsetup}). The fourth and last assumption is that the process itself happens by any unitary evolution
\begin{align}\label{unitaryevolsetupsect}
\rho'_{SR}~=~U(\rho_S\otimes\rho_R)U^\dagger~=~U\rho_{SR}U^\dagger~,
\end{align}
where the unitary $U\in\cB(\complex^{d_S}\otimes\complex^d)$ acts jointly on $S$ and $R$. Importantly, the unitarity assumption implies that no unspecified environment $E$ may participate in the process and take on entropy, which is what may happen during a dissipative process described by quantum channels. Instead, the description (\ref{unitaryevolsetupsect}) forces one to explicitly include all resources used during the process in the description by $S$ and $R$. Furthermore, if a general quantum channel were allowed for the evolution $\rho_{SR}\mapsto\rho'_{SR}$, then there obviously could not be any restriction on the entropy changes and heat flows (defined below) and no version of Landauer's bound could hold.

The unitary $U$ in (\ref{unitaryevolsetupsect}) may for example be effected by a time-dependent Schr\"odinger evolution, where an interaction Hamiltonian $H^{int}_{SR}(t)$ is switched on at some point in time by an external agent and then switched off later, such that a total unitary $U$ acts on the initial state $\rho_{SR}$ and gives the final state (\ref{unitaryevolsetupsect}). Our setup and results below do however not require any such structure for $U$; in fact, they do not even depend on the concrete unitary $U$, but merely on the fact that the evolution $\rho_{SR}\mapsto\rho'_{SR}$ was effected by \emph{some} unitary.


The joint final state $\rho'_{SR}$ of system and reservoir may be correlated, which is the generic case. By $\rho'_S:=\text{tr}_R\left[\rho'_{SR}\right]$ we denote the final state of the system (``state after the process''), by $\rho'_R:=\text{tr}_S\left[\rho'_{SR}\right]$ the final state of the reservoir. Note in particular that we do not require $\rho'_S$ to be a pure state. This assumption is frequently made in the literature, but may not always be achieved as discussed in Section \ref{boundonpureness}.

\bigskip

For a process as just described, \emph{Landauer's Principle} relates the \emph{entropy decrease of the system}
\begin{equation}
\Delta S~:=~S(\rho_S)-S(\rho'_S)\label{DeltaS}
\end{equation}
(where $S(\cdot)$ denotes the von Neumann entropy, Eq.\ (\ref{vonneumannentropynotatoinsection})) to the \emph{heat transferred to the reservoir}
\begin{equation}
\Delta Q~:=~\tr{H(\rho'_R-\rho_R)}~=~\tr{H\rho'_R}-\tr{H\rho_R}\label{DeltaQ}~,
\end{equation}
which corresponds to the (average) increase in internal energy of the thermal reservoir. The term ``heat'' is justified because this energy is not ``ordered'' since $R$ is an initially thermal reservoir, which may absorb entropy from $S$ during the process and spread the energy over many states \cite{puszworonowicz,brattelirobinson}. Another quantity of central importance will be the \emph{entropy increase of the reservoir}
\begin{align}\label{defDeltainitially}
\Delta~:=~S(\rho'_R)-S(\rho_R)~,
\end{align}
which appears in the ``Second Law Lemma'' (Lemma \ref{propsecondlaw}) and in many computations below.

In the above notation, \emph{Landauer's Principle}, which we rigorously prove below, can be written as $\beta\Delta Q\geq\Delta S$ \cite{landauer,bennettreview82,maxwelldemonbook}.

In the commonly imagined situation of ``erasure of information'' (e.g.~if the process aims to bring the system to a pure final state), the entropy of the system $S$ decreases and $\Delta S$ is positive. But unless explicitly stated otherwise, our derivations also apply to the case where the entropy of $S$ increases, corresponding to $\Delta S\leq0$, and we subsume this by the term ``process'' as well. Similarly, if heat is transferred away from the reservoir, $\Delta Q$ will be negative. Also, $\Delta$ can be negative for some processes.

\bigskip

To illustrate the above framework, we give here a paradigmatic example:
\begin{example}[Swap process]\label{swapexample}Let both the system $S$ and the reservoir $R$ have bipartite structure and agree in the dimension of one subsystem:
\begin{equation}
\complex^{d_S}~=~\complex^{d_{sw}}\otimes\complex^{d_{S_2}}~,\quad\complex^{d}~=~\complex^{d_{sw}}\otimes\complex^{d_{R_2}}\qquad(d_{sw},\,d_{S_2},\,d_{R_2}\in{\mathbb N})~.
\end{equation}
We denote the reduced states of $\rho_S$ resp.\ $\rho_R$ w.r.t.\ these bipartitions by $\rho_{S_1}$, $\rho_{R_2}$, etc. The initial system state $\rho_S$ may be chosen arbitrarily. Also, in this example, $\rho_R$ may be chosen almost arbitrarily since every given \emph{full-rank} state $\rho_R$ can be written as the thermal state of a Hamiltonian $H:=-\log\rho_R$ at finite inverse temperature $\beta:=1$.

Let the process $U:={\mathbb F}_{S_1,R_1}\otimes\1_{S_2,R_2}$ be the unitary flip ${\mathbb F}_{S_1,R_1}$ swapping the two subsystems $S_1$ and $R_1$ of dimension $d_{sw}$, i.e.\ ${\mathbb F}_{S_1,R_1}\left(\ket{\psi_1}\otimes\ket{\psi_2}\right):=\ket{\psi_2}\otimes\ket{\psi_1}$, leaving $S_2$ and $R_2$ unaffected.

For such processes, one may now compute all above quantities (\ref{DeltaS})--(\ref{defDeltainitially}) explicitly:
\begin{align}
\Delta S~&=~S(\rho_S)-S(\rho_{S_2})-S(\rho_{R_1})~,\quad&\Delta~&=~-S(\rho_R)+S(\rho_{R_2})+S(\rho_{S_1})~,\label{swapeqnrelations}\\
\Delta Q~&=~\tr{H\left(\rho_{S_1}\otimes\rho_{R_2}\,-\,\rho_R\right)}~,\quad &I(S':R')~&=~I(S_1:S_2)_{\rho_S}+I(R_1:R_2)_{\rho_R}\label{mutualinfoforswap}~,
\end{align}
where the mutual information $I(S':R')$ between $S$ and $R$ in the final state $\rho'_{SR}$ is defined in Eq.\ (\ref{definemutualinfoinnotationsection}) and will become important later.

As specific examples in the paper we will often consider processes that swap $S$ and $R$ completely, i.e.\ that have $d_S=d=d_{sw}$ and $d_{S_2}=d_{R_2}=1$. In Section \ref{attainingsection} we consider a reservoir $\C^d=(\C^{d_S})^{\otimes k}$, consisting of $k$ system copies, and a $k$-step process that swaps $S$ successively with each reservoir subsystem.
\end{example}

\subsection{Differences to previous works}\label{subsectpreworks}
Much of the early work on Landauer's Principle is based on thermodynamic reasoning rather than statistical mechanics and assumes the validity of the Second Law of Thermodynamics; this already starts with \cite{landauer} and continues e.g.\ in \cite{bennettreview82,maxwelldemonbook}. Sometimes, the content of Landauer's Principle is even taken to be the non-decrease of entropy in any process operating between a system and reservoir \cite{bennettreview82,maxwelldemonbook,nielsenchuang}. This mix of notions and the presupposition of the Second Law in the first place has caused criticism, even questioning the validity of Landauer's Principle, see e.g.\ \cite{earmannorton} and \cite{bennettnotes03} for a response. It is only in some recent derivations of Landauer's Principle that the Second Law has not been assumed explicitly, e.g.\ in \cite{shizume95,piechocinska,sagawaueda09}; we follow this line.

Our derivation of Landauer's Principle does not assume the final state $\rho'_{SR}$ to be a product state, whereas this is a common assumption in many previous treatments, especially in those using some form of the Jarzynski equality \cite{piechocinska,jarzynskiequality} or treating the quantum case (see \cite{maxwelldemonbook}). Even more specially, many derivations require a pure final system state, which however is generally not even achievable (see Section \ref{boundonpureness}). Also, we do not require the final reservoir state $\rho'_R$ to commute with the reservoir Hamiltonian $H$ \cite{piechocinska,haltasaki}. Our proofs are fully quantum-mechanical and do not require basis choices (cf.\ \cite{piechocinska}) or specific ``trajectories'' (cf.\ \cite{haltasaki} and follow-up work).

We do however explicitly assume a product initial state, $\rho_{SR}=\rho_S\otimes\rho_R$; otherwise, Landauer's Principle may be violated as discussed below Eq.\ (\ref{productstateassumptioninsetup}). Note that in derivations based on thermodynamic reasoning, the product state assumption (\ref{productstateassumptioninsetup}) is largely implicit.


Some of the literature assumes a Hamiltonian to be given also for the system $S$ (e.g. \cite{piechocinska,sagawaueda09}), and sometimes even the initial state of $S$ is supposed to be thermal. Our treatment shows that these assumptions are unnecessary; the (inverse) temperature appearing in Landauer's bound is that of the assisting reservoir. Furthermore, the assumption of a thermal initial system state would severely limit the usefulness of Landauer's Principle to information processing applications, when for example the system initially contains some outcome of a previous computation.

A Hamiltonian for the system would be necessary in order to make statements about the ``work done on the system'' (e.g.\ \cite{sagawaueda09}), which however the above stated bound $\beta\Delta Q\geq\Delta S$ does not do and which we also do not consider in this paper. In particular, the total energy of $SR$ need not be defined, but even if a Hamiltonian for $S$ were given, the total energy would typically not be conserved during the process. Our disregard of system energy is also consistent with assuming a completely degenerate system Hamiltonian (e.g.\ \cite{phfaist}).

Some treatments in the literature assume effective stochastic dynamics (e.g.\ \cite{shizume95}) instead of a unitary as in (\ref{unitaryevolsetupsect}), whereas in our formalism all involved subsystems have to be taken into account explicitly. Often in the literature the allowed dynamics is not sufficiently well specified, and this had led to the famous Maxwell's Demon paradox \cite{maxwell1871,maxwelldemonbook}; the paradox is resolved by recognizing that the Demon's memory must not be treated like an entropy sink \cite{bennettreview82} (see also Section \ref{memorysection}). Furthermore, by writing the whole process as a single discrete step as in Eq.\  (\ref{unitaryevolsetupsect}), we do not put any requirement on the ``speed'' or other structure of the evolution (as e.g.~in \cite{piechocinska,shizume95}).

Lastly, by considering the von Neumann entropy $S(\rho)=-\tr{\rho\log\rho}$ and the averaged heat transfer $\Delta Q=\tr{H(\rho'_R-\rho_R)}$, our formalism and statement of Landauer's Principle concerns an ensemble of systems: The heat transfer may be different in each single instantiation of the process, but on average equals $\Delta Q$ in the thermodynamic limit of many independent processes; similarly, $\Delta S$ is the average information decrease in the system $S$ (cf.\ beginning of Section \ref{correlationsubsection}). This asymptotic formalism is most widespread in thermodynamics and information theory, but recently \emph{one-shot statements} have been considered in these fields (e.g.\ \cite{rennerthesis}). In particular, Landauer's Principle and other work extraction statements have been formulated in the one-shot framework \cite{renneroneshotqsideinfo,trulyworklike,oppenhoro,egloff,phfaist}, but some of these derivations still rely on the asymptotic Landauer bound (\ref{landauersbasicinequality}) as one of their ingredients.

Apart from the contrast between von Neumann vs.\ one-shot entropies, note that our explicit assumption of a finite-dimensional reservoir to derive the finite-size effects in Section \ref{finitesizesect} is disjoint from the assumption of a one-shot scenario, which is sensible to address a finite (small) number of repetitions of a process.

\subsection{Notation}\label{furthernotationsubsect}
This work considers predominantly finite-dimensional systems, in particular to examine finite-size effects, and we concentrate on these here. Some notions, however, will be extended to infinite dimensions and are discussed separately at the respective places (esp.\ Appendices \ref{sectioninfinitedim} and \ref{purestateerasure}). Our treatment is for quantum systems, but contains classical systems mostly as a special case (states as probability distributions or diagonal density matrices), with exceptions noted separately. For more details on entropic quantities, see \cite{ohyapetz,wehrlreview,coverthomas,nielsenchuang}.

For any $n\in{\mathbb N}$, we denote by $\cB(\C^n)$ the set of bounded operators acting on $\C^n$, i.e.\ the $n$-dimensional complex matrix algebra. An $n$-dimensional \emph{quantum state} $\rho$ is a positive semidefinite operator $\rho\in\cB(\C^n)$ of trace one \cite{nielsenchuang}.

The \emph{(von Neumann) entropy} of a quantum state $\rho$ is defined by
\begin{align}\label{vonneumannentropynotatoinsection}
S(\rho)~:=~-\tr{\rho\log\rho}~.
\end{align}
We use the natural logarithm exclusively (denoted as $\log$), so that all entropic quantities are measured in units of \emph{nat} ($\approx1.44\,\text{bits}$). This is important to notice in Section \ref{finitesizesect} for the finite-size effects, which are non-linear in the entropies.

The \emph{relative entropy} between two states $\sigma$ and $\rho$ is defined as
\begin{align}\label{definerelentinnotationsection}
D(\sigma\|\rho)~:=~\tr{\sigma\log\sigma}-\tr{\sigma\log\rho}~.
\end{align}
It is $D(\sigma\|\rho)=\infty$ iff $\supp[\sigma]\not\subseteq\supp[\rho]$. By \emph{Klein's inequality} \cite{ohyapetz,wehrlreview,coverthomas,nielsenchuang}, the relative entropy is always non-negative and it vanishes iff $\rho=\sigma$.

For a state $\rho_{AB}$ of a bipartite system with \emph{reduced states} $\rho_A:=\text{tr}_B\left[\rho_{AB}\right]$ and $\rho_B:=\text{tr}_A\left[\rho_{AB}\right]$, the \emph{mutual information} is defined as
\begin{equation}\label{definemutualinfoinnotationsection}
I(A:B)~:=~I(A:B)_{\rho_{AB}}~:=~S(\rho_A)+S(\rho_B)-S(\rho_{AB})~,
\end{equation}
which is always non-negative. Most often it will be clear from the context for which state the mutual information is evaluated and we omit the subscript, also writing $I(A':B')$ for the mutual information between systems $A$ and $B$ in the state $\rho'_{AB}$. We sometimes use similar notation for the entropy itself, e.g.\ $S(AB):=S(\rho_{AB})$ and $S(A'):=S(\rho'_A)$.

We further define the \emph{conditional entropy (of $A$ conditioned on $B$)} in a bipartite state $\rho_{AB}$:
\begin{align}\label{defineconditionalentropy}
S(A|B)~:=~S(A|B)_{\rho_{AB}}:=S(AB)-S(B)~.
\end{align}

A \emph{Hamiltonian} $H$ on a system (which will always be the reservoir $R$ in this paper) is a Hermitian operator. The corresponding \emph{thermal state} at \emph{inverse temperature} $\beta\in[-\infty,+\infty]$ is:
\begin{equation}
\rho_\beta~:=~\frac{e^{-\beta H}}{\tr{e^{-\beta H}}}~,\label{defthermalstate}
\end{equation}
with $\rho_{\pm\infty}$ defined to be the maximally mixed state on the ground space of $\pm H$, cf.\ Eq.\ (\ref{definethermalstateinsetup}). For more details on the thermodynamics of finite-dimensional systems, see Appendix \ref{thermodynamicappendix}. For general accounts of statistical mechanics, see \cite{brattelirobinson,thirringstatphys}.

RHS and LHS abbreviate ``right-hand-side'' and ``left-hand-side'', respectively.

\section{Landauer's Principle sharpened}\label{Lprinciplesharpened}

\subsection{Equality form of Landauer's Principle}\label{Lequalitysubsection}
We now rigorously prove Landauer's Principle in the setting described in Section \ref{setupsubsection}. We actually sharpen Landauer's Principle by proving an equality version, which reduces to the famous Landauer bound \cite{landauer} when two non-negative terms are dropped. One of these two terms can be physically interpreted as some manifestation of the Second Law of Thermodynamics in our setting. We prove this first:
\begin{lemma}[Second Law Lemma]\label{propsecondlaw}
Let $\rho_{SR}=\rho_S\otimes\rho_R\in\cB(\complex^{d_S})\otimes\cB(\complex^{d})$ be a product state on a bipartite finite-dimensional system $SR$, let $U\in\cB(\complex^{d_S}\otimes\complex^{d})$ be a unitary, and denote the state after the evolution by $\rho'_{SR}=U(\rho_S\otimes\rho_R)U^\dagger$, with reduced states $\rho'_S$ and $\rho'_R$. Then:
\begin{equation}\label{propentropyproduction}
\left[S(\rho'_S)-S(\rho_S)\right]\,+\,\left[S(\rho'_R)-S(\rho_R)\right]~=~I(S':R')~\geq~0~.
\end{equation}
In the notation of Section \ref{setupsubsection}, this reads
\begin{equation}\label{shorteqnformof2law}
\Delta~=~\Delta S\,+\,I(S':R')~\geq~\Delta S~,
\end{equation}
i.e.\ the reservoir's entropy increase $\Delta$ outweighs the systems's entropy decrease $\Delta S$.
\end{lemma}
\begin{proof}Using the additivity of the von Neumann entropy for product states $\rho_{SR}=\rho_S\otimes\rho_R$ and the invariance of the entropy under the unitary evolution $\rho_S\otimes\rho_R\mapsto\rho'_{SR}$, we have:\begin{equation}
\begin{split}
S(\rho'_S)-S(\rho_S)+S(\rho'_R)-S(\rho_R)~
&=~S(\rho'_S)+S(\rho'_R)-S(\rho_S\otimes\rho_R)\\
&=~S(\rho'_S)+S(\rho'_R)-S(\rho'_{SR})\\
&=~I(S':R')_{\rho'_{SR}}~\geq~0~,
\end{split}
\end{equation}
where the last inequality follows from the non-negativity of the mutual information.
\end{proof}

The proof shows in particular that equality $\Delta=\Delta S$ in the Second Law Lemma is attained (meaning that no total entropy increase happens) iff the final state is a product state $\rho'_{SR}=\rho'_S\otimes\rho'_R$. In some accounts in the literature, the above Lemma \ref{propsecondlaw}, holding that the reservoir entropy increase exceeds the system entropy decrease, is already termed ``Landauer's Principle'' (cf.~\cite{bennettreview82,maxwelldemonbook,nielsenchuang}). We, however, take the term ``Landauer's Principle'' to mean a bound on the heat dissipation $\Delta Q$ to an initially thermal reservoir necessitated by an entropy decrease $\Delta S$ in some other system.

\bigskip

We now state and prove a sharpened version of Landauer's Principle:
\begin{theorem}[Equality form of Landauer's Principle]\label{landauereqntheorem}
Let $\rho_{SR}=\rho_{S}\otimes\rho_R\in\cB(\complex^{d_S})\otimes\cB(\complex^{d})$ be a product state on a bipartite finite-dimensional system $SR$, where $\rho_R=e^{-\beta H}/\tr{e^{-\beta H}}$ is the thermal state corresponding to a Hamiltonian $H=H^\dagger\in\cB(\complex^d)$ at inverse temperature $\beta\in[-\infty,+\infty]$. Let $U\in\cB(\complex^{d_S}\otimes\complex^{d})$ be a unitary, and denote the state after the evolution by $\rho'_{SR}=U(\rho_S\otimes\rho_R)U^\dagger$, with reduced states $\rho'_S$ and $\rho'_R$. Then:
\begin{equation}
\left[S(\rho_S)-S(\rho'_S)\right]\,+\,I(S':R')\,+\,D(\rho'_R\|\rho_R)~=~\beta\left(\tr{H\rho'_R}-\tr{H\rho_R}\right)~.\label{landauereqnlong}
\end{equation}
In the notation of Section \ref{setupsubsection}, this reads
\begin{equation}
\beta\Delta Q~=~\Delta S\,+\,I(S':R')\,+\,D(\rho'_R\|\rho_R)~,\label{landauereqnshort}
\end{equation}
which implies Landauer's bound \cite{landauer}
\begin{equation}
\beta\Delta Q~\geq~\Delta S~.\label{Lineqineqthm}
\end{equation}
\end{theorem}
\begin{proof}
First consider the case $\beta\in(-\infty,+\infty)$. Using the Second Law Lemma (Lemma \ref{propsecondlaw}) in the first line, we have
\begin{align}
\Delta S+I(S':R')~
&=~S(\rho_S)-S(\rho'_S)+I(S':R')~=~S(\rho'_R)-S(\rho_R)\label{firstlineinproofofLeq}\\
&=~-\tr{\rho'_R\log\rho'_R}+\tr{\rho_R\log\frac{e^{-\beta H}}{\tr{e^{-\beta H}}}}\nonumber\\
&=~-\tr{\rho'_R\log\rho'_R}+\tr{\rho_R\left(-\beta H-\1\log\tr{e^{-\beta H}}\right)}\nonumber\\
&=~-\tr{\rho'_R\log\rho'_R}-\beta\tr{H\rho_R}-\log\tr{e^{-\beta H}}~+~\beta\tr{H\rho'_R}-\beta\tr{H\rho'_R}\nonumber\\
&=~\beta\tr{H\left(\rho'_R-\rho_R\right)}-\tr{\rho'_R\log\rho'_R}\,+\,\tr{\rho'_R\log\frac{e^{-\beta H}}{\tr{e^{-\beta H}}}}\nonumber\\
&=~\beta\Delta Q\,-\,D(\rho'_R\|\rho_R)~.\nonumber
\end{align}

In the case $\beta=+\infty$, $\rho_R=P_g/\dim(P_g)$ is the normalised projector onto the ground state space of $H$. This implies $\tr{H\rho'_R}\geq\tr{H\rho_R}$, i.e.~$\Delta Q\geq0$. If $\Delta Q=0$, then $\rho'_R$ is supported in the ground state space as well, so that one can continue after line (\ref{firstlineinproofofLeq}) with
\begin{equation}
S(\rho'_R)-S(\rho_R)~=~-\tr{\rho'_R\log\rho'_R}-\log\dim(P_g)~=~-D(\rho'_R\|\rho_R)~,
\end{equation}
yielding that both sides of (\ref{landauereqnlong}) vanish. If $\Delta Q>0$, then $\rho'_R$ has support outside the ground state space of $H$, i.e.~outside the support of $\rho_R$, so that $D(\rho'_R||\rho_R)=+\infty$ and both sides of (\ref{landauereqnlong}) equal each other again. The reasoning in the case $\beta=-\infty$ is exactly analogous (or, alternatively, follows from the substitutions $H\mapsto-H$, $\beta\mapsto-\beta$).

Lastly, the Landauer bound (\ref{Lineqineqthm}) follows from the fact that the mutual information and the relative entropy are both non-negative.
\end{proof}



An equality equivalent to Eq.\ (\ref{landauereqnshort}) has been derived in \cite{esposito} before. There, however, the aim was to identify reversible and irreversible contributions to the entropy change, and no connection to Landauer's Principle was established. See also the ``Note added'' in Section \ref{openquestionsect}.

For extensions of Landauer's Principle (Theorem \ref{landauereqntheorem}) to infinite-dimensional separable Hilbert spaces, see Appendix \ref{sectioninfinitedim}.

\subsection{Equality cases in Landauer's bound $\beta\Delta Q\geq\Delta S$}\label{equalitysubsection}

The equality form of Landauer's Principle (Theorem \ref{landauereqntheorem}) allows us to investigate how tight the Landauer bound $\beta\Delta Q\geq\Delta S$ is (see Eq.\ (\ref{Lineqineqthm})). The basic result here is that Landauer's bound holds with equality iff, roughly speaking, the process does not do anything:
\begin{corollary}[Equality cases in the Landauer bound]\label{correqualityL}Consider a process as described in Theorem \ref{landauereqntheorem}. Then, Landauer's bound $\beta\Delta Q\geq\Delta S$ holds with equality iff there exists a unitary $V\in\cB(\complex^{d_S})$ such that
\begin{equation}\label{conditionsforequality1}
\begin{split}
~~\rho'_{SR}~&=~\rho'_S\otimes\rho'_R\\
\text{with}~~~~~\rho'_S~&=~V\rho_S V^\dagger\\
\text{and}~~~~~\rho'_R~&=~\rho_R~.
\end{split}
\end{equation}
Equivalently, Landauer's bound holds with equality iff
\begin{equation}\label{conditionsforequality2}
\Delta S~=~\Delta Q~=~0~.
\end{equation}
\end{corollary}
\begin{proof}By the equality version (\ref{landauereqnshort}) of Landauer's Principle and due to the non-negativity of the mutual information and the relative entropy, one has equality in Landauer's bound iff $I(S':R')=D(\rho'_R\|\rho_R)=0$. This is equivalent to $\rho'_{SR}$ being a product state $\rho'_{SR}=\rho'_S\otimes\rho'_R$ and $\rho'_R=\rho_R$, i.e.~to the first and third condition in (\ref{conditionsforequality1}).

This then already implies the second condition in (\ref{conditionsforequality1}) as follows. By the assumptions on the process, the states $\rho_{SR}=\rho_S\otimes\rho_R$ and $\rho'_{SR}=\rho'_S\otimes\rho_R$ before and after the process are related by a unitary transformation, $\rho'_{SR}=U\rho_{SR}U^\dagger$, and thus have the same spectra (as multisets, i.e.~including multiplicities):~${\rm{mspec}}(\rho_S\otimes\rho_R)={\rm{mspec}}(\rho'_S\otimes\rho_R)$. As the spectrum of a product state equals the pointwise product of the individual spectra, one has
\begin{equation}
{\rm{mspec}}(\rho_S)\,\cdot\,{\rm{mspec}}(\rho_R)~=~{\rm{mspec}}(\rho'_S)\,\cdot\,{\rm{mspec}}(\rho_R)~,\label{productofmspecs}
\end{equation}
and since $\rho_R$ has a non-zero eigenvalue, this implies ${\rm{mspec}}(\rho_S)={\rm{mspec}}(\rho'_S)$. So, $\rho_S$ and $\rho'_S$ are two Hermitian matrices with identical spectra, and are thus related by a unitary transformation, $\rho'_S=V\rho_S V^\dagger$ with $V\in\cB(\complex^{d_S})$.

Finally, note that the second and third condition in (\ref{conditionsforequality1}) imply $\Delta S=0$ and $\Delta Q=0$, respectively, and thus (\ref{conditionsforequality2}). Conversely, (\ref{conditionsforequality2}) obviously implies $\Delta S=\beta\Delta Q$.
\end{proof}

By Corollary \ref{correqualityL}, equality $\Delta S=\beta\Delta Q$ holds only if the process transforms the system in a unitary way and leaves the reservoir untouched, i.e.\ $\rho'_{SR}=\left(V\otimes\1_R\right)\rho_{SR}\left(V\otimes\1_R\right)^\dagger$ (note, however, that possibly $U\neq V\otimes\1_R$ when $\rho_S\otimes\rho_R$ has degenerate eigenvalues, as then the unitary transformation achieving $\rho_S\otimes\rho_R\mapsto\rho'_{SR}$ is not unique). Then there is no change in the information of the system and zero heat flow to the reservoir. In this sense, only trivial processes satisfy $\Delta S=\beta \Delta Q$; this statement remains basically true in infinite dimensions as well (Appendix \ref{sectioninfinitedim}).

\bigskip

Considering the converse implication of Corollary \ref{correqualityL}, Landauer's bound is a \emph{strict} inequality $\beta\Delta Q>\Delta S$ for any process with nonzero entropy decrease ($\Delta S\neq0$) or nonzero heat flow ($\Delta Q\neq0$). In Section \ref{finitesizesect} we will in fact derive such non-trivial lower bounds on the difference $\beta\Delta Q-\Delta S$ between the two sides of Landauer's bound (\ref{Lineqineqthm}). More precisely, we will look for a non-negative function $g(\Delta S)$ satisfying $\beta\Delta Q\geq\Delta S+g(\Delta S)$, with $g(\Delta S)>0$ for $\Delta S\neq0$; similarly, for a function $h(\Delta Q)$ such that $\beta\Delta Q\geq\Delta S+h(\Delta Q)$, with $h(\Delta Q)>0$ for $\Delta Q\neq0$.

When one fixes (or puts upper bounds on) both the system and reservoir dimensions $d_S$ and $d$, then the existence of such functions $g$ and $h$ follows because the entropy, mutual information and relative entropy are sufficiently continuous and the space of all processes as well as the state space are compact. Our functions $g$ and $h$ will indeed explicitly depend on the dimension $d$ of the reservoir. Conversely, in Section \ref{attainingsection} we show that any non-trivial $g$ or $h$ actually \emph{has to} depend on the reservoir dimension, since in the limit of large reservoir sizes $d$ we construct explicit processes coming arbitrarily close to attaining the bound $\beta\Delta Q\geq\Delta S$.


\subsection{A bound on the pureness of the final state}\label{boundonpureness}
Several discussions in the literature formulate Landauer's Principle for processes having a \emph{pure} final state $\rho'_S$, i.e.\ where the system $S$ is being brought into a \emph{definite} microstate and all information has been ``erased''. This assumption is for example made in the works \cite{piechocinska,shizume95} aiming to derive Landauer's Principle. It is also implicit in Landauer's original paper \cite{landauer} as well as in the many references that employ or ``derive'' the ubiquituous claim that an amount $(\log 2)/\beta$ of heat has to be dissipated in the ``erasure of a (qu-)bit'' (see e.g.\ several papers reprinted in \cite{maxwelldemonbook}). The latter situation would correspond to $\Delta S=\log 2$ on a system of dimension $d_S=2$, which automatically forces the final system state $\rho'_S$ to be pure, whereas the initial state $\rho_S$ must have been completely mixed.

Here we point out that a Landauer process as described above can in general \emph{not} reduce the rank of the system state $\rho_S$. This is possible only with a reservoir at strictly zero temperature or with a reservoir Hamiltonian having formally infinite energies (see below). The following impossibility result thus shows in particular that some previous statements of Landauer's Principle in the literature are void. This issue is also related to the ``unattainability formulation'' of the ``Third Law of Thermodynamics'', see also the discussions in \cite{karenarmen,nogroundstatecooling,ticozziviola}.

\bigskip

We first analyze quantitatively how the smallest eigenvalue of the system state can change during the process $\rho'_{S}={\rm tr}_R[U(\rho_S\otimes\rho_R)U^\dagger]$, described in Section \ref{setupsubsection}, using a reservoir of finite dimension $d<\infty$. Writing $\lambda_{min}(A)$ for the smallest eigenvalue of a Hermitian operator $A$, and denoting by $\ket{\psi}$ any normalized eigenvector of $\rho'_S$ corresponding to $\lambda_{min}(\rho'_S)$, we have (where $\{\ket{i_R}\}_i$ denotes any orthonormal basis for the reservoir system $R$):
\begin{align}
\lambda_{min}(\rho'_S)~&=~\bra{\psi}\rho'_S\ket{\psi}~=~\sum_{i=1}^d\bra{\psi}\bra{i_R}\,\rho'_{SR}\,\ket{\psi}\ket{i_R}\\
&\geq~\sum_{i=1}^d\lambda_{min}(\rho'_{SR})~=~d\,\lambda_{min}(\rho_S\otimes\rho_R)~=~d\,\lambda_{min}(\rho_R)\lambda_{min}(\rho_S)~.\label{lambdaminrhoprimeS}
\end{align}
Denoting by $H_{min}$ and $H_{max}$ the minimal resp.\ maximal energy (eigenvalue) of the reservoir Hamiltonian $H$, we can lower bound the minimal eigenvalue in the thermal state $\rho_R$:
\begin{align}
\lambda_{min}(\rho_R)~=~\frac{e^{-\beta H_{max}}}{\tr{e^{-\beta H}}}~\geq~\frac{e^{-\beta H_{max}}}{d\,e^{-\beta H_{min}}}~=~\frac{e^{-\beta(H_{max}-H_{min})}}{d}~\geq~\frac{e^{-2\beta\|H\|}}{d}~,
\end{align}
with the operator norm $\|H\|$ and assuming $\beta\in[0,\infty]$ (the extension to negative $\beta$ is trivial). Plugging back into (\ref{lambdaminrhoprimeS}) gives:
\begin{proposition}[Bound on the pureness of the final state]\label{propositionlambdamin}Consider any process as described in Theorem \ref{landauereqntheorem}, with a reservoir at inverse temperature $\beta\in[0,\infty]$. Then:
\begin{align}\label{boundonlambdamin}
\lambda_{min}(\rho'_S)~\geq~e^{-\beta(H_{max}-H_{min})}\lambda_{min}(\rho_S)~\geq~e^{-2\beta\|H\|}\lambda_{min}(\rho_S)~,
\end{align}
where $H_{min}$ ($H_{max}$) denotes the minimal (maximal) eigenvalue of the reservoir Hamiltonian $H$.
\end{proposition}

This means in particular that an initially full-rank state $\rho_S$ cannot be purified (cf.\ also \cite{nogroundstatecooling,ticozziviola}) unless $\beta(H_{max}-H_{min})=\infty$, i.e.~for a zero-temperature reservoir ($\beta=\infty$) or for a Hamiltonian some of whose energy levels are formally $\infty$; see also Appendices \ref{sectioninfinitedim} and \ref{purestateerasure} for the latter case. Generally, $\lambda_{min}(\rho'_S)\ll\lambda_{min}(\rho_S)$ can be achieved only by very disparate energy scales in $H$ compared to the ambient temperature $1/\beta$. The analysis leading up to Eq.~(\ref{boundonlambdamin}) also shows that $\rank{\rho'_S}\geq\rank{\rho_S}$ unless $\beta\|H\|=\infty$, since we are in finite dimensions here.

When one does allow for formally infinite energies in $H$, then for any desired $\rho'_S$ and any inverse temperature $\beta\in(0,\infty)$ one can just engineer a suitable reservoir:\ define the Hamiltonian $H:=(-\log\rho'_S)/\beta$ (using the formal convention $-\log0:=+\infty$), so that $\rho_R=\rho'_S$, and let $U$ be the process swapping $S$ and $R$ (see Example \ref{swapexample}). But note that the heat $\Delta Q=\tr{H(\rho'_R-\rho_R)}=\tr{(\rho_R-\rho'_R)\log\rho_R}/\beta$ (and not merely the norm $\|H\|$) is infinite in any process with $\rank{\rho'_R}>\rank{\rho_R}$, which necessarily happens in finite dimensions for any process achieving $\rank{\rho'_S}<\rank{\rho_S}$. At strictly zero temperature ($\beta=\infty$), similar rank-decreasing processes can be constructed without infinite $\Delta Q$.

In Appendix \ref{purestateerasure} we exhibit rank-decreasing processes at finite temperature and having finite heat flow $\Delta Q$ (and actually coming arbitrarily close to saturating the inequality $\beta\Delta Q\geq\Delta S$); such processes however need both an infinite-dimensional reservoir and formally infinite Hamiltonian levels. Note that the analysis leading up to (\ref{boundonlambdamin}) and the rank considerations above are not meaningful for infinite-dimensional reservoirs: if $\|H\|<\infty$, then the thermal state does not exist in infinite dimensions for $\beta\in[0,\infty)$ (cf.\ also Appendix \ref{sectioninfinitedim}); and when $\|H\|=\infty$, the bound (\ref{boundonlambdamin}) becomes trivial.

In cases where it is sufficient to reach a final state $\widetilde{\rho}'_S$ that is only $\delta$-close to the desired final state $\rho'_S$, i.e.\ $\|\widetilde{\rho}'_S-\rho'_S\|_1\leq\delta$, the state $\widetilde{\rho}'_S$ can be chosen to be of full rank whenever $\delta>0$. Then, from Section \ref{attainingsection} (Proposition \ref{non-rank-decreasing-prop}), one can explicitly construct a process with final state $\widetilde{\rho}'_S$ using a finite-dimensional reservoir and such that the heat dissipation $\beta\Delta Q$ is arbitrarily close to $\Delta\widetilde{S}:=S(\rho_S)-S(\widetilde{\rho}'_S)$. Note that, for given $\rho'_S$ and $\delta$, it is possible to minimize $\Delta\widetilde{S}$ by analytical methods, i.e.\ to maximize $S(\widetilde{\rho}'_S)$ subject to the constraint $\|\widetilde{\rho}'_S-\rho'_S\|_1\leq\delta$, using the Kuhn-Tucker conditions \cite{convexoptimization}. As $S(\widetilde{\rho}'_S)\geq S(\rho'_S)$ for the optimal $\widetilde{\rho}'_S$, the heat expenditure $\beta\Delta Q$ in such a process can be made arbitrarily close to $\Delta S=S(\rho_S)-S(\rho'_S)$ or smaller. Note that our impossibility results differ from the one in \cite{ticozziviola}, where it is investigated whether \emph{for all} initial states $\rho_S$ the output $\tr{U(\rho_S\otimes\rho_R)U^\dagger}$ can be $\delta$-close to a fixed pure state.

Rather than by the smallest eigenvalue, the idea behind Proposition \ref{propositionlambdamin} can be formalized via majorization theory and entropies (again in finite dimensions). Namely, when the initial states $\rho_S$, $\rho_R$  (or just their spectra) are fixed, then one of the possible final system states $\rho'_S={\rm tr}_R[U(\rho_S\otimes\rho_R)U^\dagger]$ majorizes any other such state obtained by varying $U$. The spectrum of this maximal (``purest'') state, which is unique up to unitary equivalence, is obtained by listing the $d_Sd$ eigenvalues of $\rho_S\otimes\rho_R$ in increasing order and repeatedly summing $d$ successive ones, starting from the lowest. This state has also minimal entropy $S(\rho'_S)$ among all possible final system states \cite{uhlmannnincreaseS}, but its entropy is nonzero iff $\rho_S\otimes\rho_R$ has more than $d$ nonzero eigenvalues; in particular, it is nonzero whenever $\|\beta H\|<\infty$ and $S(\rho_S)>0$.

A few treatments of Landauer's Principle in the literature do not require a pure final system state $\rho'_S$, but do assume a product final state $\rho'_{SR}=\rho'_S\otimes\rho'_R$ (such a product state would of course be implied by a pure $\rho'_S$); cf.\ e.g.\ some parts of \cite{piechocinska} (see also Section \ref{subsectpreworks}). Similar to the pure final state discussed above, also this product final state assumption is generally not achievable: A generic product state $\rho_{SR}=\rho_S\otimes\rho_R$ will admit only \emph{one} tensor product decomposition (\emph{two} when the dimensions $d_S=d$ match). Thus, the condition $\rho'_S\otimes\rho'_R=U(\rho_S\otimes\rho_R)U^\dagger$ for generic $\rho_S$, $\rho_R$ implies $U=U_S\otimes U_R$ with unitaries $U_S$, $U_R$ and so allows only trivial processes with no entropy change as $\rho'_S=U_S\rho_S U_S^\dagger$ (or, additionally, $U={\mathbb{F}}_{SR}(U_S\otimes U_R)$  in the case $d=d_S$, with the swap operator ${\mathbb{F}}_{SR}$; cf.\ Example \ref{swapexample}).

\section{Finite-size corrections to the Landauer bound}\label{finitesizesect}
The strengthened form of Landauer's Principle (Theorem \ref{landauereqntheorem}) showed that Landauer's bound $\beta\Delta Q\geq\Delta S$ is sharp only in quite trivial cases (Corollary \ref{correqualityL}). It can therefore be improved in all interesting cases. Of course, the tightest improvement is given by the equality version (\ref{landauereqnshort}), but this contains the quantities $I(S':R')$ and $D(\rho'_R\|\rho_R)$ which are usually not available as they would for example require knowledge of the full global state $\rho'_{SR}$.

In this section we derive improvements of Landauer's bound that are \emph{explicit} in the sense that they depend on the quantity $\Delta S$ that does already appear in the inequality $\beta\Delta Q\geq\Delta S$. In fact, the new bounds have to depend on the reservoir dimension $d$ as well, because processes can approach Landauer's bound $\beta\Delta Q\geq\Delta S$ in the limit $d\to\infty$ (see Section \ref{attainingsection}). The inequalities we prove in the present section thus constitute \emph{finite-size corrections} to Landauer's bound.

\bigskip

Our main result on finite-size improvements uses the following auxiliary quantities \cite{inequalityarxivpaper}:
\begin{align}
N(d)~&:=~\max_{0<r<1/2}r(1-r)\left(\log\frac{1-r}{r}(d-1)\right)^2~,\label{defineNdinproofpaper}\\
M(x,d)~&:=~\min_{0\leq s,r\leq(d-1)/d}\,\left\{\binrel(s\|r)\,\big|\,\binH(s)-\binH(r)+(s-r)\log(d-1)=x\,\right\}~,\label{definefunctionM}
\end{align}
where $2\leq d<\infty$ and $x\in[-\log d,\log d]$, and with the binary entropy $\binH(s):=S({\rm{diag}}(s,1-s))$ and the binary relative entropy $\binrel(s\|r):=D({\rm{diag}}(s,1-s)\|{\rm{diag}}(r,1-r))$. To get a better understanding of these quantities for our following main result, we remark that it follows from \cite{inequalityarxivpaper} (esp.\ Remark 4 and Lemma 14 therein) that $N(d)=\frac{1}{4}\log^2d+O(1)$ (as $d\to\infty$) and
\begin{align}\label{asymptoticM}
M(x,d)~=~\left\{\begin{array}{ll}\frac{x^2}{2N(d)}\,+\,O(x^3)&\text{as}~x\to0~(\text{for any fixed}~d\geq2)\,,\\\frac{2x^2}{\log^2d}\,+\,O\big(\frac{1}{\log^4d}\big)&\text{as}~d\to\infty~(\text{for any fixed}~x\in\R)\,.\end{array}\right.
\end{align}
Note that the quantity $(\log d)$ can be interpreted roughly as the number of particles in the reservoir $R$.

\begin{theorem}[Explicit finite-size improvements of Landauer's Principle]\label{maintheoremcombined}
Consider processes as described in Theorem \ref{landauereqntheorem}. If the reservoir dimension satisfies $2\leq d<\infty$, then
\begin{align}\label{inequalitymaintheoremfinited}
\beta\Delta Q~\geq~\left\{\begin{array}{ll}\Delta S\,+\,M(\Delta S,d)\,~\geq\,~\Delta S+\frac{(\Delta S)^2}{2N}&\quad\text{if}\,~\Delta S\geq0\\\ \\\Delta S\,+\,\left[N-\Delta S-\sqrt{N^2-2N\Delta S}\right]&\quad\text{if}\,~\Delta S\leq0\end{array}\right\}~\geq~\Delta S~,
\end{align}
for any $N\geq N(d)$ with $N(d)$ from Eq.\ (\ref{defineNdinproofpaper}); for example $N=\frac{1}{4}\log^2(d-1)+1$ or $N=\log^2d$. The function $M(\Delta S,d)$ is defined in Eq.\ (\ref{definefunctionM}). If $d=1$, then $\Delta Q=\Delta S=0$.
\end{theorem}We prove Theorem \ref{maintheoremcombined} in Section \ref{subsectDeltaSgeq0} (for $\Delta S\geq0$) and Section \ref{subsectDeltaSleq0} (for $\Delta S\leq0$). The main work for the latter case is done in Section \ref{subsectDeltaQbounds}. Note that our proofs for the two cases are quite different.

\begin{figure}
\centering
\includegraphics[trim=3.15cm 3.6cm 3.1cm 3.15cm, clip=true, scale=0.45]{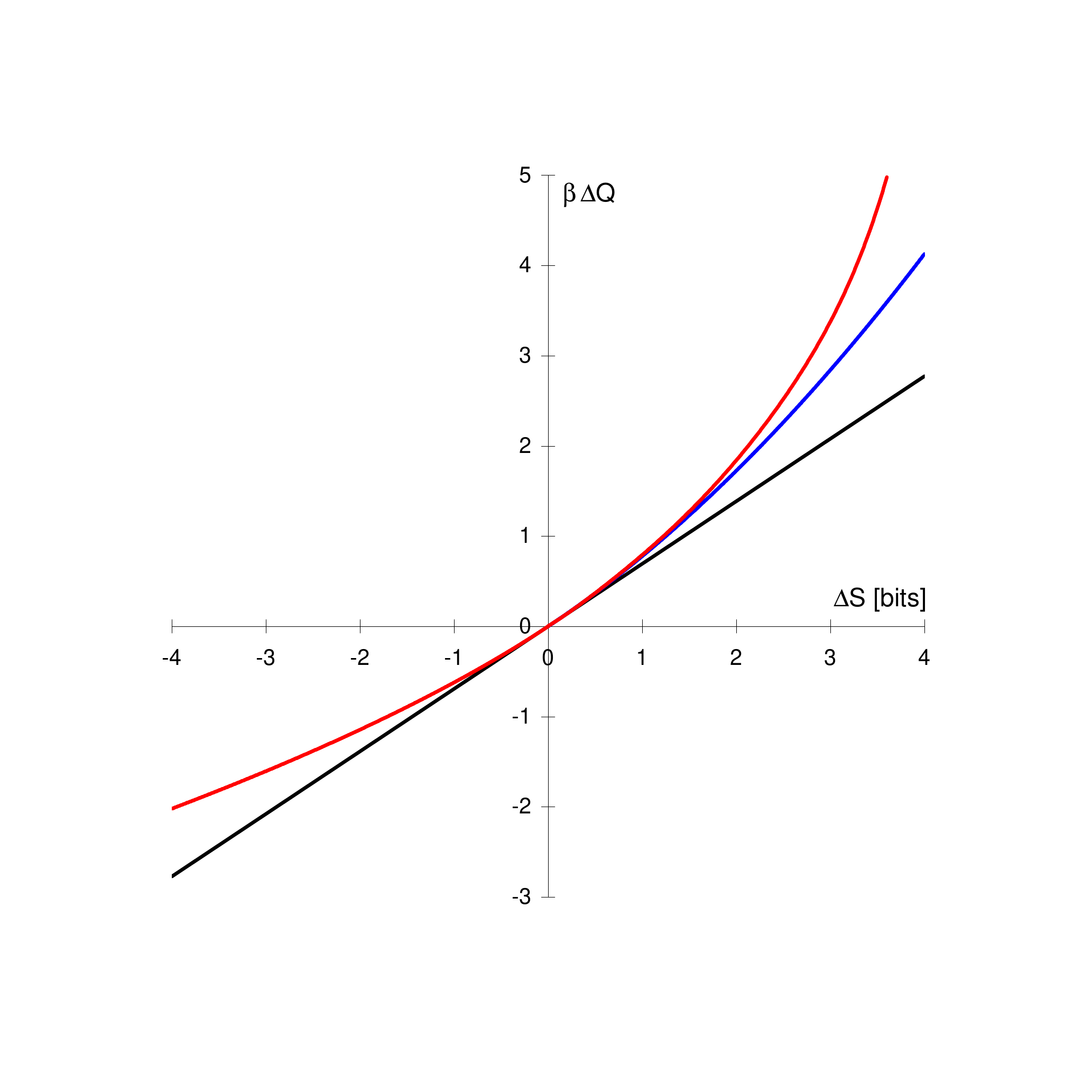}
\caption{\label{relentropygraphs}Comparison of lower bounds on the heat dissipation, for a reservoir consisting of $n=4$ qubits ($d=16$):\ Landauer's linear bound (\ref{Lineqineqthm}) is the black straight line. The red curve shows the best bounds from Eq.\ (\ref{inequalitymaintheoremfinited}) (approaching $+\infty$ as $\Delta S\to\log d$), whereas the blue curve shows the quadratic bound from Eq.\ (\ref{inequalitymaintheoremfinited}) for $\Delta S\geq0$. Points below a curve are excluded by the respective bound, i.e.\ cannot be achieved by any physical process.}
\end{figure}

The tightness of the bounds from Theorem \ref{maintheoremcombined} is investigated in Appendix \ref{tightnessofexplicitimprovementsection}. In particular we show that, for any given $d\geq2$ and $\Delta S\geq0$, the bound $\beta\Delta Q\geq\Delta S+M(\Delta S,d)$ is tight, whereas the bound (\ref{inequalitymaintheoremfinited}) for $\Delta S<0$ is not tight. Note however that at least the square brackets in (\ref{inequalitymaintheoremfinited}) is strictly positive whenever $\Delta S<0$, behaving like $(\Delta S)^2/2N+O((\Delta S)^3)$ as $\Delta S\to0$.

The lower bounds (\ref{inequalitymaintheoremfinited}) on the necessary heat dissipation $\beta\Delta Q$ are illustrated in Fig.\ \ref{relentropygraphs}. As can be seen, when erasing $1\,{\rm{bit}}$ of information assisted by a small reservoir (e.g.\ of $n=4$ qubits), then the minimum heat expenditure $\Delta Q$ necessary is tens of percent above the commonly assumed Landauer limit $\Delta S/\beta$.

\subsection{Improvement of Landauer's bound for $\Delta S\geq0$}\label{subsectDeltaSgeq0}
The main work here is accomplished by an entropy inequality proven in \cite{inequalityarxivpaper}. It gives a tight lower bound on the relative entropy $D(\sigma\|\rho)$ between any $d$-dimensional quantum states $\sigma$, $\rho$ in terms of their entropy difference $\Delta=S(\sigma)-S(\rho)$ and the dimension $d$:
\begin{align}\label{simplelookinglowerboundonrelativeentropy}
D(\sigma\|\rho)~\geq~M(\Delta,d)~,
\end{align}
where the function $M(x,d)$ is defined for $2\leq d<\infty$, $x\in[-\log d,\log d]$ in Eq.\ (\ref{definefunctionM}).

For each fixed $d$, the function $M(x,d)$ is strictly decreasing for $x\leq0$ and strictly increasing for $x\geq0$, strictly convex in $x\in[-\log d,\log d]$, and attains values $M(-\log d,d)=\log d$, $M(0,d)=0$, and $M(\log d,d)=\lim_{x\to\log d}M(x,d)=\infty$ \cite{inequalityarxivpaper}.

$M(x,d)$ can be easily computed numerically as an optimization over two bounded real variables, by its definition in Eq.\ (\ref{definefunctionM}). Furthermore, the lower bounds
\begin{align}\label{generallowerboundsonM}
M(x,d)~\geq~Ne^\frac{x}{N}-N-x~\geq~\frac{x^2}{2N}+\frac{x^3}{6N^2}~\geq~0~\qquad\forall\,x\in[-\log d,\log d]
\end{align}
hold for any $N\geq N(d)$, where $N(d)$ is defined in Eq.\ (\ref{defineNdinproofpaper}); one may for example choose $N=\frac{1}{4}\log^2(d-1)+1>N(d)$ or $N=\log^2d>N(d)$. See \cite{inequalityarxivpaper} for detailed proofs and discussion.

\bigskip

To prove (\ref{inequalitymaintheoremfinited}) in the case $\Delta S\geq0$, note that $\beta\Delta Q=\Delta+D(\rho'_R\|\rho_R)$ by the equality version of Landauer's Principle (Theorem \ref{landauereqntheorem}), where we also used $\Delta=\Delta S+I(S':R')$ by the Second Law Lemma (Lemma \ref{propsecondlaw}). The bound (\ref{simplelookinglowerboundonrelativeentropy}) gives then $\beta\Delta Q\geq\Delta+M(\Delta,d)$. Finally, Eq.\ (\ref{inequalitymaintheoremfinited}) follows for $\Delta S\geq0$ since $\Delta\geq\Delta S$ and $M(x,d)$ is monotonically increasing in $x\geq0$ \cite{inequalityarxivpaper}.

In Appendix \ref{tightnessofexplicitimprovementsection} we show that the derived inequality is tight in the sense that, for any fixed $d$ and any possible value of $\Delta S\geq0$, there exists a process, which attains equality in the bound $\beta\Delta Q\geq \Delta S+M(\Delta S,d)$.

\subsection{Improved Landauer bound depending on $\Delta Q$}\label{subsectDeltaQbounds}
The following derivation requires some notation and thermodynamics facts from Appendix \ref{thermodynamicappendix}, where these things are proven rigorously. We also use the presuppositions and notation from Theorem \ref{maintheoremcombined} (made explicit in the statement of Theorem \ref{landauereqntheorem}).

Denote by $E_R:=\tr{H\rho_R}=\tr{H\rho_\beta}=E(\beta)$ the initial energy of the reservoir, and denote by $\rho'_{R,th}=e^{-\beta' H}/\tr{e^{-\beta' H}}$ the thermal state of the same energy as $\rho'_R$, i.e.\ $\tr{H\rho'_R}=E_R+\Delta Q={\rm{tr}}[H\rho'_{R,th}]=E(\beta')$. Here, $\beta'\in[-\infty,+\infty]$ is uniquely determined if $H\not\propto\ii_R$ (Appendix \ref{thermodynamicappendix}), which we assume from now on; otherwise it would be $\Delta Q=0$, in which case the final result Theorem \ref{sharpenLineqforDeltaQleq0tm} follows directly from Theorem \ref{landauereqntheorem}.

As is easily seen by using the thermal form of $\rho_R$ and $\rho'_{R,th}$, the relative entropy term from (\ref{landauereqnshort}) can be rewritten as follows (this is ``Pythagoras' Theorem'' first noticed in \cite{csiszarpythagoras}):
\begin{align}\label{pythagoreantheoremforD}
D(\rho'_R\|\rho_R)~=~D(\rho'_R\|\rho'_{R,th})+D(\rho'_{R,th}\|\rho_R)~.
\end{align}
Note that the first term on the RHS is always finite, whereas the other two terms are always both finite or both infinite. Rewriting the last term due to thermality of $\rho_R$ and Lemma \ref{lemmaSEbeta} gives
\begin{align}
D(\rho'_R\|\rho_R)~&=~D(\rho'_R\|\rho'_{R,th})\,+\,\tr{(\rho'_{R,th}-\rho_R)(-\log\rho_R)}-\left[S(\rho'_{R,th})-S(\rho_R)\right]\label{firstlinewithaddends}\\
&=~D(\rho'_R\|\rho'_{R,th})\,+\,\beta\Delta Q\,-\,\int_{E_R}^{E_R+\Delta Q}\beta(E)\,dE~.\label{firsttimeintegralappears}
\end{align}
Except possibly for the term ${\rm{tr}}[\rho'_{R,th}(-\log\rho_R)]$ (and consequently $\beta\Delta Q$) in case $\beta=\pm\infty$, all addends in (\ref{firstlinewithaddends})--(\ref{firsttimeintegralappears}) are finite (even though the integrand can diverge at the boundaries when $\beta,\beta'=\pm\infty$). Treating this case with the usual conventions, we can continue:
\begin{align}
D(\rho'_R\|\rho_R)~
&=~D(\rho'_R\|\rho'_{R,th})\,+\,\int_{E_R}^{E_R+\Delta Q}\left(\beta(E_R)-\beta(E)\right)\,dE\\
&=~D(\rho'_R\|\rho'_{R,th})\,+\,\int_{E_R}^{E_R+\Delta Q}\left(-\int_{E_R}^E\frac{d\beta(E')}{dE'}\,dE'\right)\,dE\\
&=~D(\rho'_R\|\rho'_{R,th})\,+\,\int_{E_R}^{E_R+\Delta Q}\int_{E_R}^E\frac{1}{\var_{\beta(E')}(H)}\,dE'\,dE\label{lastequationincontinuousderviationwithequality}~,
\end{align}
where in the last step we used Eq.~(\ref{dbetadE}).

Notice that always $\beta(E')\in[\beta,\beta']$ in the last integral, where we understand the notation $[\beta,\beta']$ for $\beta'<\beta$ to mean the interval $[\beta',\beta]$. We can thus (even in the case $\Delta Q<0$) put a lower bound on the double integral by replacing the denominator by $\max_{\gamma\in[\beta,\beta']}\var_\gamma(H)$. Furthermore dropping the relative entropy term gives
\begin{equation}
D(\rho'_R\|\rho_R)~\geq~\int_{E_R}^{E_R+\Delta Q}\int_{E_R}^E\frac{1}{\max_{\gamma\in[\beta,\beta']}\var_\gamma(H)}\,dE'\,dE~=~\frac{(\Delta Q)^2}{2\max_{\gamma\in[\beta,\beta']}\var_\gamma(H)}~.\label{lowerboundonDwithDeltaQsquared}
\end{equation}

\bigskip

We aim for a lower bound on $D(\rho'_R\|\rho_R)$ that involves the quantity $\beta\Delta Q$, which already appears in the usual Landauer bound, rather than $\Delta Q$ alone; at the same time we would like to eliminate the complicated expression in the denominator, which resembles a heat capacity (cf.\ (\ref{heatcapacitybeta}) and (\ref{heatcapacityT}) in Appendix \ref{thermodynamicappendix}). To do this, assume first $\beta\in(-\infty,+\infty)\setminus\{0\}$ to get
\begin{align}
D(\rho'_R\|\rho_R)~\geq~\frac{(\beta\Delta Q)^2}{2\max_{\gamma\in[\beta,\beta']}\beta^2\var_\gamma(H)}~.\label{aftererweiternmitbeta}
\end{align}

If $\Delta Q\leq0$, then $\beta'\geq\gamma\geq\beta$, since the energy is strictly decreasing with the inverse temperature. Thus, if $\Delta Q\leq0$ and $\beta>0$, the denominator in (\ref{aftererweiternmitbeta}) can be upper bounded by $2\max_{\gamma\in[\beta,\beta']}\gamma^2\var_\gamma(H)$. The same holds for $\beta<0$ and $\Delta Q\geq0$, since then we have $\beta'\leq\gamma\leq\beta<0$. If $\beta\Delta Q\leq0$ and $\beta\in(-\infty,+\infty)\setminus\{0\}$, we thus have
\begin{align}
D(\rho'_R\|\rho_R)~\geq~\frac{(\beta\Delta Q)^2}{2\max_{\gamma\in[\beta,\beta']}\var_\gamma(\gamma H)}~.\label{aftererweiternmitgamma}
\end{align}
When $\gamma=\pm\infty$, the expression $\var_\gamma(\gamma H)$ should be interpreted as $0$, since $\beta^2\var_{\gamma=\pm\infty}(H)=0$ in (\ref{aftererweiternmitbeta}) due to $\beta\neq\pm\infty$.

The key observation is now that $\var_\gamma(\gamma H)$ can be upper bounded \emph{independently} of the Hamiltonian $H$ or the inverse temperature $\gamma$ just as a function of the reservoir dimension $d$ (see (\ref{heatcapacityT}) and below). In fact, the following tight bound was proven in \cite{inequalityarxivpaper}
\begin{align}\label{upperboundonvargamma}
\var_\gamma(\gamma H)~\leq~N(d)~<~\frac{1}{4}\log^2(d-1)+1~\qquad\text{for}~2\leq d<\infty~,
\end{align}
where $N(d)$ is defined in (\ref{defineNdinproofpaper}). This also holds for $\gamma=\pm\infty$, due to the convention from the previous paragraph; to see this, note that the LHS of (\ref{upperboundonvargamma}) is written in \cite{inequalityarxivpaper} as $\var_{\rho_\gamma}(\log\rho_\gamma)$, and this equals $0$ when $\rho_\gamma$ is maximally mixed on its support subspace, in particular for $\gamma=\pm\infty$.

We thus have:
\begin{align}\label{simplelowerboundonbetaDwithNdandDeltaQ}
D(\rho'_R\|\rho_R)~\geq~\frac{(\beta\Delta Q)^2}{2N(d)}~\qquad\text{if}~\,\beta\Delta Q\leq0~.
\end{align}
This statement holds also for $\beta\in\{-\infty,0,+\infty\}$, which was not included in the above derivation. This is because $\beta=+\infty$ necessitates $\Delta Q\geq0$, which together with the condition $\beta\Delta Q\leq0$ enforces $\Delta Q=0$, so that the numerator in (\ref{simplelowerboundonbetaDwithNdandDeltaQ}) vanishes and the inequality holds. Similarly, the numerator vanishes for $\beta=-\infty,0$. Note that the RHS of (\ref{simplelowerboundonbetaDwithNdandDeltaQ}) does never diverge as long as $\beta\Delta Q\leq0$, because of $\beta\Delta Q\geq\Delta S+I(S':R')=\Delta\geq-\log d$ due to (\ref{landauereqnshort}) and (\ref{shorteqnformof2law}).

\bigskip

Using Eq.\ (\ref{simplelowerboundonbetaDwithNdandDeltaQ}) in Theorem \ref{landauereqntheorem}, and also using Lemma \ref{propsecondlaw}, we finally arrive at:
\begin{theorem}[Sharpening of Landauer's bound for $\beta\Delta Q\leq0$]\label{sharpenLineqforDeltaQleq0tm}Consider processes as described in Theorem \ref{landauereqntheorem}, with the reservoir dimension satisfying $2\leq d<\infty$. If the initial inverse temperature $\beta\in[-\infty,+\infty]$ and the heat dissipation $\Delta Q$ satisfy $\beta\Delta Q\leq0$, then:
\begin{align}
\Delta S~\leq~\Delta S+I(S':R')~=~\Delta~\leq~\beta\Delta Q-\frac{(\beta\Delta Q)^2}{2N(d)}~\qquad\text{if}\,~\beta\Delta Q\leq0~,\label{inequalityinDeltaQleq0thm}
\end{align}
where $N(d)$ is defined in Eq.\ (\ref{defineNdinproofpaper}).
\end{theorem}

The right inequality in (\ref{inequalityinDeltaQleq0thm}) is generally wrong if one does not demand $\beta\Delta Q\leq0$, because for any $d\geq2$ it is easy to construct Hamiltonians $H$ such that $\beta\Delta Q$ becomes arbitrarily large (positive, but finite), so that the RHS in (\ref{inequalityinDeltaQleq0thm}) becomes arbitrarily negative, whereas $\Delta\geq-\log d$ is bounded from below.

The derivation leading up to Theorem \ref{sharpenLineqforDeltaQleq0tm} shows how more detailed knowlege about the reservoir (i.e.\ about the temperature, the Hamiltonian, or its heat capacity) could be exploited, when available, to obtain better bounds than (\ref{simplelowerboundonbetaDwithNdandDeltaQ}) or (\ref{inequalityinDeltaQleq0thm}). With knowledge of only the reservoir dimension $d$, however, the essential bound (\ref{upperboundonvargamma}) is tight \cite{inequalityarxivpaper}. Bounds similar to (\ref{simplelowerboundonbetaDwithNdandDeltaQ}) or (\ref{inequalityinDeltaQleq0thm}) are possible also in the case $\beta\Delta Q>0$ if one for example has a lower bound on $|\beta'|$, i.e.\ if one knows by how much the temperature can rise at most by the addition of the heat amount $\Delta Q$.

\subsection{Improvement of Landauer's bound for $\Delta S\leq0$}\label{subsectDeltaSleq0}
Landauer's bound $\beta\Delta Q\geq\Delta S$ does not forbid values of $\beta\Delta Q$ close to $\Delta S$ (see Eq.\ (\ref{Lineqineqthm})). In the case $\Delta S<0$ it thus ``allows'' some negative values of $\beta\Delta Q$. But then Theorem \ref{sharpenLineqforDeltaQleq0tm} gives new constraints and we will use these to prove (\ref{inequalitymaintheoremfinited}) in the case $\Delta S\leq0$. Assume therefore a process with $\Delta S\leq0$ throughout this section.

If $\beta\Delta Q\geq0$, then the inequality (\ref{inequalitymaintheoremfinited}) holds since
\begin{align}
\Delta S+\left[N-\Delta S-\sqrt{N^2-2N\Delta S}\right]~=~N-\sqrt{N^2-2N\Delta S}~\leq~0~\leq~\beta\Delta Q~,
\end{align}
due to $\Delta S\leq0$ and $N\geq N(d)\geq0$.

Assume therefore now $\beta\Delta Q<0$ (as noted below (\ref{simplelowerboundonbetaDwithNdandDeltaQ}), it is $\beta\Delta Q>-\infty$ always). In this case, we use Theorem \ref{sharpenLineqforDeltaQleq0tm},
\begin{equation}
\Delta~\leq~\beta\Delta Q-\frac{(\beta\Delta Q)^2}{2N(d)}~,\label{intermediateeqninDeltaSleq0proof}
\end{equation}
multiply this by $2N(d)$, and rearrange to get
\begin{equation}
\left(N(d)-\beta\Delta Q\right)^2~\leq~N(d)^2-2N(d)\Delta~.
\end{equation}
This implies $\beta\Delta Q\geq N(d)-\sqrt{N(d)^2-2N(d)\Delta}$ and, via $\Delta S\leq\Delta$ due to Lemma \ref{propsecondlaw},
\begin{align}
\beta\Delta Q~\geq~N(d)-\sqrt{N(d)^2-2N(d)\Delta S}~.
\end{align}

The last expression only decreases when $N(d)$ is replaced by any $N\geq N(d)$, as one verifies easily. This finally proves inequality (\ref{inequalitymaintheoremfinited}) in the case $\Delta S\leq0$.

\section{Landauer processes involving correlations}\label{extendednotionssection}

More general than in Section \ref{setupsubsection}, we consider in Section \ref{memorysection} a setup where initial correlations may be used during the process. In Section \ref{correlationsubsection} we ask for thermodynamic constraints on the erasure of correlations themselves (rather than entropy). Further extensions of the basic setup are described in Appendix \ref{extendednotionssectionapp}.

\subsection{Landauer processes with memory and noisy operations}\label{memorysection}
More generally than in the setup from Section \ref{setupsubsection}, the agent who aims to modify (e.g.\ to ``erase'') the system's initial state $\rho_S$ may have some information about the actual microstate (e.g.\ pure state) of the system $S$. In this case, the desired process may be accomplished with less heat expenditure than given by naive application of Landauer's bound $\beta\Delta Q\geq\Delta S$ (see e.g.\ \cite{renneroneshotqsideinfo}). Formally, this additional knowledge can be described through an additional memory system $M$ that may initially be correlated with the system $S$, and such that the unitary $U$ may now act jointly on all three systems $S$, $R$, and $M$.

For example, when the system state is $\rho_S=\sum_ip_i\ket{i}_S\bra{i}$ (with orthonormal states $\{\ket{i}_S\}$) and the agent had perfect classical knowledge about the microstate $\ket{i}$ of $S$, the situation would be described by $\rho_{SM}=\sum_ip_i\ket{i}_S\bra{i}\otimes\ket{i}_M\bra{i}$, whereas perfect quantum correlation would correspond to a pure (entangled) initial state $\ket{\psi}_{SM}=\sum_i\sqrt{p_i}\ket{i}_S\ket{i}_M$ of system and memory. In both examples, if the process is a unitary $U=U_{SM}\otimes\ii_R$ acting non-trivally only on $SM$ in such a way that $U_{SM}\left(\ket{i}_S\ket{i}_M\right)=\ket{\psi}_S\ket{i}_M$ (with any fixed pure state $\ket{\psi}_S$), then one easily verifies
\begin{align}\label{changeofstates}
\rho'_{SRM}~=~\big(U_{SM}\otimes\ii_R\big)\rho_{SRM}\big(U^\dagger_{SM}\otimes \ii_R\big)~=~\ket{\psi}_S\bra{\psi}\otimes\rho'_{RM}\qquad\text{with}~\,\rho'_R=\rho_R~,
\end{align}
i.e.\ the information from $S$ is completely erased ($S(\rho'_S)=0$), whereas $R$ remains unchanged (in the first example above with initially perfect classical correlations also $M$ remains unchanged, $\rho'_{RM}=\rho_{RM}$); in particular, no entropy or heat increase occurs in the reservoir, $\Delta=\Delta Q=0$. This seems to contradict the Second Law Lemma (Eq.\ (\ref{shorteqnformof2law})) and Landauer's bound (Eq.\ (\ref{Lineqineqthm})), but is of course due to the initial correlations with $M$ that the process $U$ can access.

Further extending the setup from Section \ref{setupsubsection}, instead of only unitary interactions $U$ (Eq.\ (\ref{unitaryevolsetupsect})), one may allow for so-called ``noisy operations'' \cite{noiseoperation}, i.e.\ unitaries using an additional completely mixed ancilla system, or more generally any unital quantum channel $T$. For this, we use that a unital positive and trace-preserving map $T$ does not decrease the entropy \cite{uhlmannnincreaseS}.

\bigskip

The above points motivate the following setup, which extends the one from Section \ref{setupsubsection} and to which we can easily generalize our treatment:
\begin{enumerate}
\item[(a')]\vspace{-0.15cm}the system $S$, reservoir $R$, and memory $M$ are initially in a joint quantum state $\rho_{SRM}$,
\item[(b')]\vspace{-0.15cm}the initial reduced reservoir state $\rho_R={\rm tr}_{SM}\left[\rho_{SRM}\right]$ is thermal, $\rho_R=e^{-\beta H}/\tr{e^{-\beta H}}$,
\item[(c')]\vspace{-0.15cm}the process proceeds by a unital positive trace-preserving map $T$, i.e.\ $\rho'_{SRM}=T(\rho_{SRM})$,
\item[(d')]\vspace{-0.15cm}the entropy and heat changes $\Delta S$, $\Delta$, $\Delta Q$ are defined on the marginal states as in Fig.\ \ref{figsetup}.
\end{enumerate}

\bigskip

A modified Second Law Lemma (cf.\ Lemma \ref{propsecondlaw}) for this more general situation is then immediately verfied:
\begin{align}
\Delta~&=~S(R')-S(R)~=~\left[S(SRM)-S(R)\right]-\left[S(SRM)-S(R')\right]\nonumber\\
&\geq~\left[S(SRM)-S(R)\right]-\left[S(S'R'M')-S(R')\right]\label{ineqduetounitaleqn}\\
&=~\left[S(SM)-I(SM:R)\right]-\left[S(S'M')-I(S'M':R')\right]\nonumber\\
&=~\left[S(S|M)-S(S'|M')\right]\,+\,I(S'M':R')\,+\,\left[S(M)-S(M')\right]\,-\,I(SM:R)~,\label{verygeneralsecondlawwithmemory}
\end{align}
where $S(S|M)$ from (\ref{defineconditionalentropy}) is the entropy of $S$ conditioned on $M$. The inequality in (\ref{ineqduetounitaleqn}) is due to $S(\rho'_{SRM})=S(T(\rho_{SRM}))\geq S(\rho_{SRM})$ \cite{uhlmannnincreaseS} and will be an equality if $T$ is unitary.

If one only considers processes where the memory register $M$ is not being altered (as e.g.\ in \cite{renneroneshotqsideinfo}), implying $S(M')\leq S(M)$, and where the reservoir was initially uncorrelated with the rest, $\rho_{SRM}=\rho_{SM}\otimes\rho_R$ (see Section \ref{setupsubsection}), then one still has
\begin{equation}\label{secondlawcondeqn}
\Delta~\geq~\Delta S_{cond}\,+\,I(S'M':R')~\geq~\Delta S_{cond}\qquad\text{with}~~\Delta S_{cond}:=S(S|M)-S(S'|M')~,
\end{equation}
similar to (\ref{shorteqnformof2law}). Intuitively it is clear that $\Delta\geq\Delta S_{cond}$ need not hold when either the memory $M$ takes on some of the entropy, i.e.\ when $S(M')>S(M)$, or when the initial total entropy was reduced due to correlations with $R$, i.e.\ when $I(SM:R)>0$; both possibilities constitute resources that may be exploited for more efficient processes. Note that the Second Law Lemma just outlined in (\ref{verygeneralsecondlawwithmemory})--(\ref{secondlawcondeqn}) does not require a thermal state $\rho_R$ nor a Hamiltonian $H$ for the reservoir; but when the reservoir is initially thermal (see condition (b) above) then it is natural to assume no initial reservoir correlations, $I(SM:R)=0$, see Section \ref{setupsubsection} and \cite{relativethermalization}.

Under the assumptions $\rho_{SRM}=\rho_{SM}\otimes\rho_R$ and $S(M')\leq S(M)$ (in addition to (a')--(d') above), one arrives thus at the following form of Landauer's Principle, generalizing Eq.\ (\ref{landauereqnshort}),
\begin{align}\label{landauersgeneralizedineq}
\beta\Delta Q~\geq~\Delta S_{cond}\,+\,I(S'M':R')\,+\,D(\rho'_R\|\rho_R)~\geq~\Delta S_{cond}~.
\end{align}
The proof is as in (\ref{firstlineinproofofLeq}), but now starting from (\ref{secondlawcondeqn}). All the finite-size improvements from Section \ref{finitesizesect} apply to this more general case with memory as well if only $\Delta S$ is replaced by $\Delta S_{cond}$.

One can evaluate all above statements for the two examples given around Eq.\ (\ref{changeofstates}). In both cases, $I(S'M':R')=I(SM:R)$, $S(M)=S(S)$, and $S(S'|M')=0$. Furthermore, for the classically correlated case the initial conditional entropy was $S(S|M)=0$ and the state of the memory did not change, $S(M')=S(M)=S(S)$, whereas in the case of maximal quantum correlations $S(S|M)=-S(S)$ is negative and the final memory state is pure $S(M')=0$. The latter case is the most interesting: The generalization (\ref{landauersgeneralizedineq}) of Landauer's Principle is not tight in this case, since the memory state was purified at the expense of the quantum correlations between $S$ and $M$; a subsequent unitary interaction between $M'$ and $R'$ may however reduce the reservoir energy to give in the end $\beta\Delta Q=-S(M)=-S(S)$ and $S(M'')=S(M)$.

As a final remark, if there is no memory system $M$ but possibly initial correlations in $SR$ \cite{alickicorrelations}, then (\ref{verygeneralsecondlawwithmemory}) can be written as
\begin{align}
\Delta~=~S(S)-S(S')+I(S':R')-I(S:R)~=~S(S|R)-S(S'|R')~\geq~\Delta S-I(S:R)~.
\end{align}
Now one can formulate a Landauer Principle in terms of the difference $S(S|R)-S(S'|R')$ rather than $\Delta S_{cond}$ as above; or alternatively, one can bound the mutual information term $I(S:R)$, which appears with the ``wrong'' sign, by more traditional quantities like the trace distance, $I(S:R)\leq\|\rho_{SR}-\rho_S\otimes\rho_R\|_1(\log d_S+\log d_R)$, and this gives corrections to the usual Landauer bound (similarly for the term $I(SM:R)$ in (\ref{verygeneralsecondlawwithmemory})). Using similar processes as above with $M$ (around Eq.\ (\ref{changeofstates})), one can see that for initially perfect classical or quantum correlations in $SR$, one can achieve $\Delta S=S(S)$ while still $\Delta Q=0$ (due to $\rho'_R=\rho_R$); this ``violation'' of Landauer's bound is of course explained by $\rho_{SR}\neq\rho_S\otimes\rho_R$, contrary to the assumption (\ref{productstateassumptioninsetup}).

\subsection{A Landauer Principle for correlations rather than entropies?}\label{correlationsubsection}
The common formulation of Landauer's Principle \cite{landauer} says that changing the information in a system (e.g.\ by ``erasing information'') puts constraints on the heat dissipated during the process. This statement is consistent with the mathematical content of Theorem \ref{landauereqntheorem} when the \emph{entropy} $S(\rho)$ is interpreted as the \emph{amount of information} in a system in state $\rho$, and thus $\Delta S=S(\rho_S)-S(\rho'_S)$ is interpreted as the decrease of information in $S$. Such an interpretation of entropy is substantiated by the fundamental theorems of asymptotic information theory \cite{shannonmathemtheorycommun,schumachercoding}.

This interpretation of entropy also corresponds to the situation where the system has been prepared by someone in any one of the (orthonormal, and thus perfectly distinguishable) pure states $\ket{i}\bra{i}$ according to the distribution $\left\{p_i\right\}_i$, such that however the index $i$ is unknown to a second agent (who thus describes the system state as $\rho=\sum_ip_i\ket{i}\bra{i}$; see \cite{bennettnotes03} for further discussion). In this sense, data or information is contained \emph{in} the system and may be retrieved by the second agent through a measurement in the basis $\left\{\ket{i}\right\}_i$; this measurement yields the information $\sum_ip_i(-\log p_i)=S(\rho)$ on average over many independent retrievals. Mathematically, $\sum_ip_i(-\log p_i)$ is the minimum (over all complete measurements) of the averaged measurement outcome information on a state with eigendecomposition $\rho=\sum_ip_i\ket{i}\bra{i}$.

In contrast to the information stored \emph{in} a system, which was just quantified by the entropy, one can instead consider the information someone has \emph{about} a system. The information that an agent (with memory register $M$) has about the state of system $S$ is simply the \emph{correlations between $S$ and $M$}, described by the joint state $\rho_{SM}$ of the combined system $SM$ (see also Section \ref{memorysection}). And the \emph{amount} of correlations between $S$ and $M$ is quantified by the mutual information $I(S:M)$ (again, for an averaged or asymptotic scenario \cite{shannonmathemtheorycommun}). This makes sense since $I(S:M)=0$ is equivalent to $\rho_{SM}=\rho_S\otimes\rho_M$, meaning that the agent's memory does not hold any information about the microstate of $S$, whereas $I(S:M)=S(S)$ iff $S(S|M)=0$, such that the agent has (on average) perfect classical knowledge about the state of $S$.

\bigskip

One may now wonder whether a version of Landauer's bound also holds for the change of \emph{information about} a system. We show here that a straightforward analogy does not work. For the setup assume that, besides an initially thermal reservoir $R$ that is uncorrelated with the other systems (see Sections \ref{setupsubsection} and \ref{memorysection}), there are a system $S$ and a memory register $M$, which may be correlated:
\begin{align}
\rho_{SRM}~=~\rho_{SM}\otimes\rho_R\,,\qquad\rho_R=\frac{e^{-\beta H}}{\tr{e^{-\beta H}}}~.
\end{align}
The information about $S$ is thus $I(S:M)$ initially. Then the system $S$ and reservoir $R$ is subjected to a joint unitary process as described in Section \ref{setupsubsection},
\begin{align}
\rho'_{SRM}~=~\big(U_{SR}\otimes\ii_M\big)\rho_{SRM}\big(U^\dagger_{SR}\otimes\ii_M\big)~,
\end{align}
and we examine how the information of the memory $M$ about the system $S$ changes:
\begin{align}
\Delta I~:=~I(S:M)-I(S':M')~.
\end{align}
The process imagined does not affect the memory $M$; if it were allowed to, then $\Delta I$ can be virtually independent of the heat change $\Delta Q$, so that no version of Landauer's bound (such as possibly $\beta\Delta Q\geq\Delta I$) can hold. Note further that it is always $\Delta I\geq0$ in such processes due to the data processing inequality \cite{nielsenchuang,coverthomas}; this corresponds to ``information erasure'', whereas the entropy change $\Delta S$ in Sections \ref{setupsubsection} and \ref{memorysection} could have either sign.

But even so, there cannot be a straightforward version of Landauer's bound involving $\Delta I$. To see this, take any state $\rho_{SM}$, and consider a reservoir $R$ of the same size as $S$ and with initial state $\rho_R=\rho_S:={\rm{tr}}_M\left[\rho_{SM}\right]$ (note that every full-rank state $\rho_R$ is the thermal state of some Hamiltonian $H:=-\log\rho_R$ at $\beta=1$). Let the process $U_{SR}:={\mathbb{F}}_{SR}$ be the swap of $S$ and $R$ (cf.\ Example \ref{swapexample}). Then $I(S':M')=0$, and so $\Delta I=I(S:M)$, since $R$ and $M$ were initially uncorrelated, whereas $\beta\Delta Q=0$ due to $\rho'_R=\rho_R$. Thus, the tentative inequality $\beta\Delta Q\geq\Delta I$ is here violated whenever $I(S:M)>0$. The latter happens in particular when $M$ initially has perfect classical or quantum knowledge about a non-pure state $\rho_S$ (cf.\ beginning of Section \ref{memorysection}).

As another violating example, take a product initial state $\rho_{SRM}=\rho_S\otimes\rho_R\otimes\rho_M$ with $d:={\rm{dim}}(R)={\rm{dim}}(S)\geq2$ and again $U_{SR}={\mathbb{F}}_{SR}$; then $\Delta I=0$, whereas $\beta\Delta Q=\tr{(\rho_R-\rho_S)(\log\rho_R)}$ may assume either sign. Namely, $\beta\Delta Q$ becomes negative for example when $\rho_S:=\ket{\psi}\bra{\psi}$ is any pure state and $\rho_R:=(1-\lambda)\ket{\psi}\bra{\psi}+\lambda\ii/d$ with any $\lambda\in(0,1)$, since (cf.\ Remark 5 in \cite{inequalityarxivpaper})
\begin{align}
\beta\Delta Q~&=~\tr{(\rho_R-\rho_S)\log\rho_R}~=~S(\rho_S)-S(\rho_R)+D(\rho_S\|\rho_R)\\
&=~-S\left((1-\lambda)\ket{\psi}\bra{\psi}+\lambda\ii/d\right)+D\left(\ket{\psi}\bra{\psi}\,\|\,(1-\lambda)\ket{\psi}\bra{\psi}+\lambda\ii/d\right)\nonumber\\
&<~-(1-\lambda)S\left(\ket{\psi}\bra{\psi}\right)-\lambda S\left(\ii/d\right)+(1-\lambda)D\left(\ket{\psi}\bra{\psi}\,\|\,\ket{\psi}\bra{\psi}\right)+\lambda D\left(\ket{\psi}\bra{\psi}\,\|\,\ii/d\right)\nonumber\\
&=~\lambda D\left(\ket{\psi}\bra{\psi}\,\|\,\ii/d\right)\,-\,\lambda S\left(\ii/d\right)~=~\lambda\left(\log d-\log d\right)~=~0\nonumber
\end{align}
due to strict concavity of the entropy and convexity of the relative entropy; one can actually find $\lambda$ such that $\beta\Delta Q<-0.4\log d$, whereas $\beta\Delta Q>0.2-\log d$ for any $\rho_S$, $\rho_R$. Again, the inequality $\beta\Delta Q\geq\Delta I$ does not hold here. On the other hand, $\beta\Delta Q$ is positive by Theorem \ref{landauereqntheorem} for any $\rho_S$, $\rho_R$ with $S(\rho_S)>S(\rho_R)$ and can be come arbitrarily big for any fixed $d\geq2$; thus, also a reversed inequality, such as tentatively $\beta\Delta Q\leq\Delta I$, cannot hold in general.

\bigskip

Other tentative notions of a Landauer Principle for correlations can be dismissed similarly. One may for example define \emph{complete erasure of information} to mean any process $U=U_{SR}\otimes\ii_M$, together with a thermal resource state $\rho_R$, which satisfies
\begin{align}
{\rm{tr}}_R\left[U(\psi_{SM}\otimes\rho_R)U^\dagger\right]~=~\frac{\ii_S}{d_S}\otimes\frac{\ii_M}{d_S}\qquad\forall~\text{max.~entangled or class.~correlated}~\psi_{SM}~.
\end{align}
Such a complete erasure process is necessarily a swap of $S$ with a $d_S$-dimensional completely mixed subsystem of $R$. But this does not require any heat dissipation, as shown in the first example above where $\beta\Delta Q=0$.

\section{Processes approaching Landauer's bound}\label{attainingsection}
Theorem \ref{landauereqntheorem} is a sharpened version of Landauer's Principle, and Theorem \ref{maintheoremcombined} makes the sharpening  more explicit through dimension-dependent lower bounds on the improvement. Given this, one may now wonder about the possibility for dimension-independent improvements of the Landauer bound $\beta\Delta Q\geq\Delta S$ \cite{landauer}. To answer this, we construct here processes which, for a desired state transformation $\rho_S\mapsto\rho'_S$, approach Landauer's bound arbitrarily closely. This is analogous to processes on single systems which come close to extracting the maximal amount of work allowed by the Second Law from a nonequilibrium system, see e.g.\ \cite{andersgiovannetti,Skrzypczyk,catalyticcoherence}.

By Section \ref{boundonpureness}, a process $\rho_S\mapsto\rho'_S$ is achievable with a finite-dimensional reservoir only if $\rank{\rho'_S}\geq\rank{\rho_S}$; this is the case we treat below, formulating our construction as Proposition \ref{non-rank-decreasing-prop}. The following construction also illustrates that, for any $\Delta S\neq0$, the reservoir dimension has to grow indefinitely as Landauer's bound $\beta\Delta Q\geq\Delta S$ is approached (see Theorem \ref{maintheoremcombined}). Rank-decreasing processes are the subject of Appendix \ref{purestateerasure}.

\begin{proposition}[Rank-non-decreasing processes]\label{non-rank-decreasing-prop}Let two quantum states $\rho_S,\rho'_S\in\cB(\C^{d_S})$ be given with $1\leq d_S<\infty$ and $\rank{\rho'_S}\geq\rank{\rho_S}$, and let $\varepsilon>0$. Then there exists a reservoir of finite dimension $d<\infty$ with Hamiltonian $H\in\cB(\C^d)$ and inverse temperature $\beta:=1$ and a unitary $U$, such that the resulting process (see Section \ref{setupsubsection}) satisfies
\begin{align}
\beta\Delta Q~\leq~\Delta S+\varepsilon~.
\end{align}
That is, Landauer's bound $\beta\Delta Q\geq\Delta S$ can be approached arbitrarily closely.
\end{proposition}
\begin{proof}Denote $r:=\rank{\rho'_S}$. We construct the reservoir $R$ as consisting of $k$ subsystems $\C^r$, i.e.\ $d=d_S^k$, and the reservoir Hamiltonian as a sum of local Hamiltonians $H=\sum_{i=1}^kH_i$, where each $H_i$ acts nontrivially only on subsystem $i$. The initial thermal reservoir state is thus $\rho_R=\otimes_{i=1}^k\rho_R^{(i)}$, where $\rho_R^{(i)}=e^{-H_i}/\tr{e^{-H_i}}$ are the local thermal states. We will construct the unitary $U$ as a product of several unitaries (``stepwise process''). As the zeroth step, apply a unitary $U_0$ to the system $S$ alone such that $\supp[U_0\rho_S U_0^\dagger]\subseteq\supp[\rho'_S]$; this does not change any entropies or cause any heat flow. For any state $\rho$ on $S$, denote by $\rho|_r$ the $r$-dimensional restriction onto the support of $\rho'_S$.

We now define $\rho_0:=U_0\rho_S U_0^\dagger$, $\rho_k:=\rho'_S$, and choose intermediate states $\rho_i$ satisfying $\supp[\rho_i]=\supp[\rho'_S]$ for $k=1,\ldots,k$. One possible choice is \cite{andersgiovannetti}
\begin{align}\label{andersgiovannettiprescription}
\rho_i~=~\left(1-\frac{i}{k}\right)\rho_0+\frac{i}{k}\rho'_S\qquad(i=0,\ldots,k)~,
\end{align}
but one can choose any $k-1$ points along a curve $\rho(t)$ ($t\in[0,1]$) connecting $\rho_S$ and $\rho'_S$ in the space of states in such a way that $\rho(t)$ is supported on the full subspace $\supp[\rho'_S]$ for all $t>0$. Define then the local Hamiltonians $H_i:=-\log(\rho_i|_r)\in\cB(\C^r)$; this gives $\rho_R^{(i)}=\rho_i|_r$.

\begin{figure}
\centering
\includegraphics{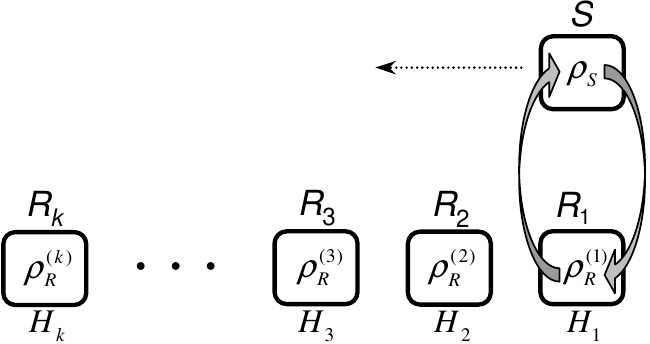}
\caption{\label{successiveswaps}The $k$-step process constructed in the proof of Proposition \ref{non-rank-decreasing-prop} (shown here for the case of full-rank $\rho_S$ and $\rho'_S$): The first step $\rho_S\equiv\rho_R^{(0)}\leftrightarrow\rho_R^{(1)}$ swaps the initial state of $S$ with the state of the first subsystem of $R$. After the $k$-th swap, $S$ is in the desired state $\rho'_S$, since the reservoir Hamiltonian (and temperature) has been constructed such that in thermal equilibrium the $k$-th reservoir subsystem has state $\rho_R^{(k)}=\rho'_S$. In the limit $k\to\infty$ of many small steps, the heat production exceeds $(S(\rho_S)-S(\rho'_S))/\beta=\Delta S/\beta$ only by a vanishingly small amount.}
\end{figure}
Finally, define the unitary $U_i$ in step $i$ to be the operator swapping the reservoir subsystem $i$ with the $r$-dimensional subspace $\supp[\rho'_S]$ of $S$, see also Fig.\ \ref{successiveswaps}.

After $k$ steps, the final system state is thus $\rho'_S$ and the system entropy has changed by $\Delta S:=S(\rho_S)-S(\rho'_S)$. The heat dissipation is (denoting $\rho_R^{(0)}:=\rho_S|_r$):
\begin{align}
\beta\Delta Q~&=~\tr{\beta H\,(\rho'_R-\rho_R)}~=~\tr{(\rho_R-\rho'_R)\log\rho_R}\\
&=~\tr{\left(\otimes_{i=1}^k\rho_R^{(i)}\,-\,\otimes_{i=1}^k\rho_R^{(i-1)}\right)\log\left(\otimes_{i=1}^k\rho_R^{(i)}\right)}\\
&=~\sum_{i=1}^k\tr{\left(\rho_R^{(i)}-\rho_R^{(i-1)}\right)\log\rho_R^{(i)}}~=~\sum_{i=1}^k\tr{(\rho_i-\rho_{i-1})\log\rho_i}\label{laststepbeforeriemannsum}~.
\end{align}
We now take any fixed curve $\rho(t)$, as outlined above, and make the discretization $\rho_i:=\rho(t_i)$ finer as $k\to\infty$, as in the definition of the Riemann integral (e.g.\ as in (\ref{andersgiovannettiprescription})). Then (\ref{laststepbeforeriemannsum}) equals
\begin{align}
\sum_{i=1}^k\tr{\left(\rho(t_i)-\rho(t_{i-1})\right)\log\rho(t_i)}~,
\end{align}
which for $k\to\infty$ converges to
\begin{align}
\to~&~\int_{t=0}^1\tr{d\rho(t)\,\log\rho(t)}~=~\int_0^1dt\,\tr{\dot{\rho}(t)\,\log\rho(t)}\\
&~~=~-\int_0^1dt\frac{d}{dt}\tr{\rho(t)-\rho(t)\log\rho(t)}~=~-\tr{\rho(1)-\rho(0)}\,-\,S\left(\rho(1)\right)+S\left(\rho(0)\right)\\
&~~=~S(\rho_S)-S(\rho'_S)~=~\Delta S~.
\end{align}
Thus, for any $\varepsilon>0$, there exists $k\in\nat$ such that for the associated process $\beta\Delta Q\leq\Delta S+\varepsilon$.
\end{proof}

For any fixed value of $k$ in the preceding proof we can also write, by Theorem \ref{landauereqntheorem} (Eq.\ (\ref{landauereqnshort})),
\begin{align}\label{betaDeltaQforksteps}
\beta\Delta Q~=~\Delta S+D(\rho'_R\|\rho_R)~=~\Delta S+\sum_{i=1}^{k}D(\rho_{i-1}\|\rho_i)~,
\end{align}
since $I(S':R')=0$ due to the swap processes (cf.\ Example \ref{swapexample}). In \cite{andersgiovannetti} an upper bound is derived for the sum in (\ref{betaDeltaQforksteps}) when using the prescription (\ref{andersgiovannettiprescription}):
\begin{align}
\sum_{i=1}^{k}D(\rho_{i-1}\|\rho_i)~\leq~\sum_{i=1}^{k}\left(D(\rho_{i-1}\|\rho_i)+D(\rho_i\|\rho_{i-1})\right)\,~\stackrel{\text{Eq.~}(\ref{andersgiovannettiprescription})}{=}\,~\frac{D(\rho_0\|\rho'_S)+D(\rho'_S\|\rho_0)}{k}~,\label{andersgiovannettiupperbound}
\end{align}
which is explicitly seen to converge to $0$ for $k\to\infty$ when ${\rm rank}[\rho'_S]={\rm rank}[\rho_S]$.

Conversely, there is a \emph{lower} bound on $\beta\Delta Q-\Delta S$ for \emph{any} $k$-step process, due to the convexity of the function $M(x,r)$ from Eq.\ (\ref{simplelookinglowerboundonrelativeentropy}) in its first argument (see \cite{inequalityarxivpaper}):
\begin{align}
\beta\Delta Q-\Delta S~=~D(\rho'_R\|\rho_R)~=~\sum_{i=1}^kD(\rho_R^{(i-1)}\|\rho_R^{(i)})~\geq~k\,M(\Delta S/k,r)~\geq~k\,M(\Delta S/k,d_S)~.\label{explicitlowerboundongapwithM}
\end{align}
This is stronger than the direct bound $D(\rho'_R\|\rho_R)\geq M(\Delta S,r^k)\geq M(\Delta S,d_S^k)$ from Eq.\ (\ref{simplelookinglowerboundonrelativeentropy}), since it is $M(\Delta S,r^k)=O((\Delta S)^2/k^2)$ as $k\to\infty$, whereas the RHS of (\ref{explicitlowerboundongapwithM}) is $O((\Delta S)^2/k)$ \cite{inequalityarxivpaper}. This shows that the $O(1/k)$ convergence in (\ref{andersgiovannettiupperbound}) for the prescription (\ref{andersgiovannettiprescription}) is optimal for stepwise processes.

The number of steps $k$ above may be interpreted as the time duration of the whole process, assuming that each individual swap consumes constant time. Thus, lower bounds on the difference $\beta\Delta Q-\Delta S$ as in (\ref{explicitlowerboundongapwithM}) corroborate the folklore that processes can become reversible only in the limit of slow processes and large process times (see also Section 3.1.1 in \cite{inequalityarxivpaper}).

\section{Open questions}\label{openquestionsect}
In the present work, we have investigated the energy expediture $\Delta Q$ necessary to decrease the system entropy by $\Delta S$ and have improved Landauer's bound (\ref{landauersbasicinequality}) in case of a small reservoir. With increasing technological control over miniature systems, the actual initial and desired final states $\rho_S$ and $\rho'_S$ (or their minimum eigenvalues, etc.) become more important than just their entropy difference. Thus, one can ask for better bounds $\beta\Delta Q\geq\Delta S+f(\rho_S,\rho'_S,d)$ than those implied by Theorem \ref{maintheoremcombined}; strictly better bounds do exist for generic $\rho_S$, $\rho'_S$, cf.\ Appendix \ref{tightnessofexplicitimprovementsection} and \cite{inequalityarxivpaper}. Such an improvement could also provide a stabilized version of the result from Section \ref{boundonpureness}, according to which $\rank{\rho'_S}<\rank{\rho_S}$ implies $\beta\Delta Q=\infty$ (when $d<\infty$). The explicit constructions from Section \ref{attainingsection} however show that non-trivial functions $f$ must necessarily depend on the finite reservoir size $d<\infty$.

A related issue is to find tight finite-size improvements in Theorem \ref{maintheoremcombined} for $\Delta S<0$. If a function $g$ is to satisfy $\beta\Delta Q\geq\Delta S+g(\Delta S,d)$, then the explicit examples from Appendix \ref{tightnessofexplicitimprovementsection} give upper bounds on it: $g(\Delta S,d)\leq M(\Delta S,d)$ for $\Delta S\in[-\log d,0]$, and $g(\Delta S,d)\leq-\Delta S$ for $\Delta S\in[-2\log d,-\log d]$. Our work however leaves open the question whether the best possible $g$ is smaller than those values.

Going beyond the setting of separable Hilbert spaces (Appendix \ref{sectioninfinitedim}), Landauer's Principle can probably be formulated within the general statistical mechanical framework of $C^*$- or $W^*$-dynamical systems \cite{puszworonowicz,brattelirobinson,thirringstatphys}, and an equality version akin to (\ref{landauereqnshort}) can possibly be proven along the lines of Theorem \ref{landauereqntheorem}. Note that in this framework the mutual information can be written as a relative entropy and the heat flow as a derivation w.r.t.\ the dynamical semigroup. For a finite-dimensional system $S$ and general operator-algebraic reservoir $R$, the recent preprint \cite{jaksic} gives a generalization of the equality from \cite{esposito} (see ``Note added'' below).

Finally, one may wonder whether thermodynamics puts constraints also on the erasure of correlations between systems. Straightforward guesses at such relations, inspired by Landauer's Principle, were shown in Section \ref{correlationsubsection} to be violated in general. Also, one-shot formulations (e.g.\ in the framework of \cite{renneroneshotqsideinfo,trulyworklike,oppenhoro,phfaist}) of our equality version of Landauer's Principle and the explicit finite-size corrections remain for future work.

\medskip{\bf Note added.} After completion of the present work, we became aware of the paper \cite{esposito} by Esposito, Lindenberg, and Van den Broeck, which for a setting similar to Section \ref{setupsubsection} gives an equality that is easily seen to be equivalent to Eq.\ (\ref{landauereqnshort}) above. While these authors identify the mutual information and relative entropy terms as the ``irreversible entropy production'' during the process and illustrate recurrences under continuous time-evolution in an explicit model, they do not make any connections to Landauer's Principle and in particular do not identify why the presuppositions of Theorem \ref{landauereqntheorem} are reasonable to capture this scenario (see Section \ref{setupsubsection}). No general explicit finite-size improvements of Landauer's bound were given, nor have achieving processes been discussed there.

\medskip{\bf Acknowledgments.} We thank Francesco Buscemi, Philippe Faist, Geza Giedke, Patrick Hayden, Steve Hsu, Daniel Lercher, Fernando Pastawski, and Marco Piani for discussions and valuable suggestions. DR acknowledges support from the Marie Curie Intra European Fellowship QUINTYL and the COST Action MP1209 ``Thermodynamics in the quantum regime''. MMW was supported by the Alfried Krupp von Bohlen und Halbach-Stiftung.

\appendix
\section{Thermodynamics of finite-dimensional systems}\label{thermodynamicappendix}
Here we collect and rigorously prove some facts from the thermodynamics of finite-dimensional systems, which are necessary especially to derive the finite-size improvements in Section \ref{finitesizesect}. In the main text of this paper, the reservoir $R$ plays the role of the $d$-dimensional system below (cf.\ Section \ref{setupsubsection}). All following derivations are equally valid for classical systems with finite state space (note that all occuring thermal states are diagonal in the eigenbasis of the Hamiltonian). For a more general development of statistical mechanics, see e.g.\ \cite{thirringstatphys,brattelirobinson}.

For the setup, we need a quantum system of finite Hilbert space dimension $d$, $1\leq d<\infty$, and a given and fixed \emph{Hamiltonian} $H$ for this system, i.e.\ a Hermitian operator $H\in\cB(\C^d)$.

Then, for any \emph{inverse temperature} $\beta\in[-\infty,+\infty]$, the corresponding \emph{thermal state} is
\begin{align}\label{thermalstateinappendix}
\rho_\beta~:=~\frac{e^{-\beta H}}{\tr{e^{-\beta H}}}~\qquad(\beta\in[-\infty,+\infty])~,
\end{align}
with the convention that $\rho_{\pm\infty}$ denotes the maximally mixed state on the ground space of $\pm H$. (The latter convention is physically sensible, and furthermore ensures $\lim_{\beta\to\pm\infty}\rho_\beta=\rho_{\pm\infty}$, so that $\rho_\beta$ is continuous in $\beta\in[-\infty,+\infty]$.) Thermal states are often only defined for non-negative $\beta$, but for finite-dimensional systems there is no necessity for this restriction, besides occasional notational convenience. Physically speaking, thermal states are the stable (``equilibrium'') states of a system with Hamiltonian $H$ at temperature $1/\beta$ \cite{puszworonowicz,brattelirobinson,thirringstatphys}. As such, they are ``cheaply available'' when these physical conditions are met and can be used ``at no cost'' during the processes described in Section \ref{setupsubsection}.

We denote the \emph{thermal average} of an operator $A\in\cB(\C^d)$ sometimes by $\langle A\rangle_\beta:=\tr{A\rho_\beta}$ and its variance by $\var_\beta(A):=\langle \left(A-\langle A\rangle_\beta\right)^2\rangle_\beta$.

The \emph{energy $E=E(\beta)$ of a thermal state} is the thermal average of the Hamiltonian:
\begin{align}\label{energythermal}
E(\beta)~:=~\tr{H\rho_\beta}~=~\tr{H\frac{e^{-\beta H}}{\tr{e^{-\beta H}}}}~\qquad(\beta\in[-\infty,+\infty])~.
\end{align}
Obviously, $E(\beta)$ is a continuous function of $\beta\in[-\infty,+\infty]$ and smooth in $\beta\in(-\infty,+\infty)$. By continuity at $\beta=\pm\infty$ we mean $\lim_{\beta\to\pm\infty}E(\beta)=E(\pm\infty)$. It is easy to see that $E(\infty)=E_{min}(H)$ and $E(-\infty)=E_{max}(H)$, where $E_{min}(H)$ and $E_{max}(H)$ denote the minimal and maximal eigenvalues (\emph{energy levels}) of $H$, respectively.

\begin{lemma}[Heat capacity]\label{strictlydecreasingElemma}
Let $H$ be a Hamiltonian on a finite-dimensional system. Then:
\begin{equation}
\frac{d}{d\beta}E(\beta)~=~-\var_\beta(H)\qquad\text{for}~~\beta\in(-\infty,+\infty)~.\label{heatcapacitybeta}
\end{equation}

If $H$ has at least two distinct energy levels, i.e.\ if $H\not\propto\1_R$, then the energy $E(\beta)=\tr{H\rho_\beta}$ is strictly decreasing in $\beta\in[-\infty,+\infty]$, with strictly negative derivative $\frac{d}{d\beta}E(\beta)~<~0$ for $\beta\in(-\infty,+\infty)$.

If $H\propto\1_R$, then $E(\beta)=\tr{H}/d$ is constant in $\beta\in[-\infty,+\infty]$.
\end{lemma}
\begin{proof} For $\beta\in(-\infty,\infty)$,
\begin{align}
\frac{d}{d\beta}\tr{H\frac{e^{-\beta H}}{\tr{e^{-\beta H}}}}~&
=~-\tr{H^2\frac{e^{-\beta H}}{\tr{e^{-\beta H}}}}+\left(\tr{H\frac{e^{-\beta H}}{\tr{e^{-\beta H}}}}\right)^2\\
&=~-\tr{\rho_\beta\left(H-\tr{H\rho_\beta}\1_R\right)^2}~=~-\var_\beta(H) ,
\end{align}
which is strictly negative if $H\neq\tr{H\rho_\beta}\1_R$ since $\rho_\beta$ is of full rank. Continuity $\lim_{\beta\to\pm\infty}E(\beta)=E(\pm\infty)$ gives then strict monotonicity in $\beta\in[-\infty,+\infty]$. The case $H\propto\ii_R$ is obvious.
\end{proof}Lemma \ref{strictlydecreasingElemma} implies that the inverse of the function $E=E(\beta)$, namely
\begin{align}
\beta:[E_{min}(H),E_{max}(H)]\to[-\infty,+\infty]\,,~~\beta~=~\beta(E)~,\label{betafunctionofE}
\end{align}
exists iff $H\not\propto\1_R$, is strictly decreasing, continuous, and smooth in the interior of its domain, with derivative
\begin{align}
\frac{d}{dE}\,\beta(E)~=~\left(\left.\frac{dE(\beta)}{d\beta}\right|_{\beta=\beta(E)}\right)^{-1}~=~\frac{-1}{\var_{\beta(E)}(H)}~.\label{dbetadE}
\end{align}

\bigskip

The LHS of (\ref{heatcapacitybeta}) is also called the \emph{heat capacity w.r.t.\ inverse temperature}, i.e.~the instantaneous rate of change in the system energy as $\beta$ is varied; by Lemma \ref{strictlydecreasingElemma} this equals the (negative) energy fluctuations. When thermal states are parametrized in terms of the \emph{temperature} $T:=1/\beta$ (let here $T\in(-\infty,+\infty)\setminus\{0\}$), then the corresponding \emph{heat capacity} is
\begin{align}
C(T)~:=~\frac{d}{dT}E(T)~=~\frac{d\beta(T)}{dT}\cdot\frac{dE(\beta)}{d\beta}~=~\beta^2\var_\beta(H)~=~\var_\beta(\beta H)~.\label{heatcapacityT}
\end{align}
Note that this quantity equals $\var_{\beta}(\log\rho_\beta)$, the variance of the operator $\log\rho_\beta$ in the thermal state $\rho_\beta$. In Section \ref{subsectDeltaQbounds}, we use that (\ref{heatcapacityT}) is upper bounded just in terms of the system dimension, namely by $N(d)$, which is defined in Eq.\ (\ref{defineNdinproofpaper}), cf.\ (\ref{upperboundonvargamma}). See \cite{inequalityarxivpaper} for a proof.

\bigskip

The \emph{entropy $S=S(\beta)$ of a thermal state} is (cf.\ also Section \ref{furthernotationsubsect})
\begin{equation}
S(\beta)~:=~-\tr{\rho_\beta\log\rho_\beta}~=~-\tr{\frac{e^{-\beta H}}{\tr{e^{-\beta H}}}\log\frac{e^{-\beta H}}{\tr{e^{-\beta H}}}}~\qquad(\beta\in[-\infty,+\infty])~.\label{Sfunctionofbeta}
\end{equation}
Essentially paralleling the discussion following (\ref{energythermal}), $S(\beta)$ is continuous in $\beta\in[-\infty,+\infty]$ (by \cite{fannesinequalitypaper}), and smooth in the interior of its domain with first derivative (after some elementary computation)
\begin{equation}
\frac{d}{d\beta}S(\beta)~=~-\beta\var_\beta(H)~\qquad(\beta\in(-\infty,+\infty))~.\label{derivativeSbeta}
\end{equation}

If $H\not\propto\1_R$, the temperature $\beta=\beta(E)$ is a function of the energy by (\ref{betafunctionofE}), and (with some common abuse of notation) the entropy $S(E):=S(\beta(E))$ can also be viewed as a function of the energy $E$ of a thermal state. This function is well-defined even in the case $H\propto\1_R$, since then $\rho_\beta=\rho_{\beta'}$ for all $\beta,\beta'\in[-\infty,+\infty]$. Thus one always has a well-defined function
\begin{equation}
S:[E_{min}(H),E_{max}(H)]\to[0,\log d]\,,~S~=~S(E)~.\label{SfunctionofE}
\end{equation}

The entropy, energy and temperature of thermal states are related as follows:
\begin{lemma}[Relation between $S$, $E$, and $\beta$ for thermal states]\label{lemmaSEbeta}
Let $H$ be a Hamiltonian on a finite-dimensional system. Then the entropy $S=S(E)$ as a function of the energy of thermal states is continuous, smooth in the interior of its domain, and has first derivative
\begin{equation}
\frac{d}{dE}S(E)~=~\beta(E)\qquad\text{for}~~E_{min}(H)<E<E_{max}(H)~.
\end{equation}
\end{lemma}
\begin{proof}If $H\propto\1_R$, the domain of $S(E)$ consists of a single point and there is nothing to prove. Otherwise, since $\beta=\beta(E)$ is smooth and strictly decreasing in the interior of the domain by (\ref{betafunctionofE})--(\ref{dbetadE}) and $S=S(\beta)$ is smooth in the interior of its domain by (\ref{Sfunctionofbeta})--(\ref{derivativeSbeta}), the smoothness claim follows; similarly does continuity on the whole domain. Then, using the chain rule and Eqs.~(\ref{dbetadE}) and (\ref{derivativeSbeta}) for $E\in(E_{min}(H),E_{max}(H))$:
\begin{equation}
\frac{d}{dE}S(E)~=~\left.\frac{dS(\beta)}{d\beta}\right|_{\beta=\beta(E)}\cdot\left.\frac{d\beta(E)}{dE}\right|_E~=~\frac{-\beta(E)\var_{\beta(E)}(H)}{-\var_{\beta(E)}(H)}~=~\beta(E)~.
\end{equation}
\end{proof}

\bigskip

Lastly, we introduce the free energy, a notion used ubiquitously in traditional thermodynamics. The relative entropy between any state $\rho$ and a thermal state $\rho_\beta$ can, for $\beta\in(-\infty,\infty)$, be written as
\begin{align}
D(\rho\|\rho_\beta)~&=~\beta\left(\tr{H\rho}-\tr{H\rho_\beta}\right)\,-\,\left(S(\rho)-S(\rho_\beta)\right)\label{freeenergydiffasD}\\
&=~\beta F_\beta(\rho)\,-\,\beta F_\beta(\rho_\beta)~,\label{writeasdifferenceoffreeenergies}
\end{align}
where we defined the \emph{dimensionless free energy $\beta F_\beta$} of any state $\rho$ as:
\begin{align}
\beta F_\beta(\rho)~:=~\beta\tr{H\rho}-S(\rho)~.\label{defineBetaF}
\end{align}
With obvious and usual conventions, Eqs.\ (\ref{freeenergydiffasD})--(\ref{defineBetaF}) hold for all $\beta\in[-\infty,+\infty]$ (note that (\ref{defineBetaF}) need not equal $0$ for $\beta=0$). Applied to $D(\rho\|\rho_\beta)$ in Eqs.\ (\ref{freeenergydiffasD})--(\ref{writeasdifferenceoffreeenergies}), Klein's inequality (see below Eq.\ (\ref{definerelentinnotationsection})) gives several versions of the \emph{thermodynamic inequality}: The thermal state $\rho=\rho_\beta$ is the unique maximizer of the entropy at fixed energy, and (for $\beta\geq0$) is the unique minimizer of the energy at fixed entropy. Equivalently, the functional $\beta F_\beta(\rho)$ is uniquely minimized by $\rho=\rho_\beta$, which corresponds to the usual free energy minimization in thermodynamics (for $\beta\geq0$). See \cite{ohyapetz,thirringstatphys,brattelirobinson} for more detailed discussions.

\section{Tightness of the finite-size improvements}\label{tightnessofexplicitimprovementsection}
Here we investigate how tight our finite-size bounds from Section \ref{finitesizesect} are.

Let $\rho_{SR}=\rho_S\otimes\rho_R\mapsto\rho'_{SR}=U(\rho_S\otimes\rho_R)U^\dagger$ be any process as considered in Theorem \ref{maintheoremcombined} (see also Section \ref{setupsubsection}), with a reservoir of dimension $d<\infty$. Before discussing the tightness of the bound Eq.\ (\ref{inequalitymaintheoremfinited}), we investigate the range of possible values of the quantity $\Delta S$, on which the bound depends. 

When $d$ is fixed, one can put upper and lower bounds on the entropy change $\Delta S$ of the system. A lower bound is obtained by
\begin{align}
\Delta S~&=~S(S)-S(S')~=~S(S)-[S(S'R')+I(S':R')-S(R')]\\
&=~S(S)-[S(S)+S(R)+I(S':R')-S(R')]~=~-S(R)-[I(S':R')-S(R')]~,
\end{align}
where we used $S(S'R')=S(RS)=S(S)+S(R)$ by unitarity (\ref{unitaryevolsetupsect}) and the product initial state assumption (\ref{productstateassumptioninsetup}). Now, for quantum systems $I(S':R')\leq2S(R')$, whereas the stronger inequality $I(S':R')\leq S(R')$ holds for classical systems \cite{wehrlreview,ohyapetz,nielsenchuang}. Lower bounds on $\Delta S$ are then obtained by noting $S(R),S(R')\leq\log d$:
\begin{equation}\label{upperlowerboundsonDeltaS}
\begin{aligned}
-2\log d~&\leq~\Delta S~\leq~\log d~~~\qquad\text{(quantum systems)\,,}\\
-\log d~&\leq~\Delta S~\leq~\log d~~~\qquad\text{(classical systems)\,,}
\end{aligned}
\end{equation}where for the upper bounds we used $\Delta S\leq\Delta=S(R')-S(R)\leq\log d$ by Lemma \ref{propsecondlaw}.

All inequalities can be attained when only the reservoir dimension $d$ is fixed: A swap (see Example \ref{swapexample}) between a pure $\rho_R$ and a maximally mixed $\rho_S$ (of dimension $d_S=d$) attains both upper bounds, whereas swapping a maximally mixed $\rho_R$ with a pure $\rho_S$ attains the classical lower bound. For the quantum lower bound, take the system $S$ to be composed of two $d$-dimensional subsystems $S_1$, $S_2$ in a maximally entangled initial state $\ket{\omega}_{S_1,S_2}$, and the reservoir $R$ again initially maximally mixed. Then the process that swaps $S_1$ and $R$ creates the final state $\rho'_{SR}=(\ii_{S_1}/d)\otimes\ket{\omega}_{S_2,R}\bra{\omega}$ with a maximally entangled state $\ket{\omega}_{S_2,R}$. Thus, $\rho'_S=\ii_S/d^2$ so that $\Delta S=-2\log d$. In this example, $\rho'_R=\rho_R=\ii_R/d$, which means there is no heat flow, $\Delta Q=0$.

\bigskip

We now investigate how tight the inequality (\ref{inequalitymaintheoremfinited}) from Theorem \ref{maintheoremcombined} is. Specifically, for any given $d$ and $\Delta S$ (which are the quantities appearing in the bound), does there exist a process such that the lower bound (\ref{inequalitymaintheoremfinited}) on $\beta\Delta Q$ holds with equality? The answer is in the affirmative when $\Delta S\geq0$ (which by (\ref{upperlowerboundsonDeltaS}) means $0\leq\Delta S\leq\log d$), but not for $\Delta S<0$.

To see this, consider a swap process (Example \ref{swapexample}) between a system with $d_S=d$ dimensions and the $d$-dimensional reservoir. Due to $I(S':R')=0$ and $\rho'_R=\rho_S$, $\rho'_S=\rho_R$, Eq.\ (\ref{landauereqnshort}) gives:
\begin{align}\label{optimalswapprocess}
\beta\Delta Q~=~\Delta S+D(\rho_S\|\rho_R)~.
\end{align}
Now, by \cite{inequalityarxivpaper}, for any given $d$ (with $2\leq d<\infty$) and given $\Delta S\in[-\log d,+\log d]$ there do exist $d$-dimensional states $\rho_S$, $\rho_R$ with $S(\rho_S)-S(\rho_R)=\Delta S$ and $D(\rho_S\|\rho_R)=M(\Delta S,d)$ (one can choose $\rho_R$ as a thermal state of some Hamiltonian at finite temperature, except when $\Delta S=\log d$, which requires $\beta=\infty$). The swap process with these special initial states $\rho_S$, $\rho_R$ thus satisfies $\beta\Delta Q=\Delta S+M(\Delta S,d)$, which shows that for $\Delta S\geq0$ the best lower bound on $\beta\Delta Q$ from (\ref{inequalitymaintheoremfinited}) is tight. This tightness holds for quantum as well as for classical systems, as the states $\rho_S,\rho_R$ satisfying $D(\rho_S\|\rho_R)=M(\Delta S,d)$ commute \cite{inequalityarxivpaper}. We leave open the question whether for $\Delta S>0$ the bound can be exactly tight even if one fixes $d_S$ independently of $d$.

For $\Delta S<0$ however, the lower bound on $\beta\Delta Q$ is given by the second selector in (\ref{inequalitymaintheoremfinited}), for which one can prove by using the leftmost inequality in (\ref{generallowerboundsonM}):
\begin{align}
\Delta S+M(\Delta S,d)~>~\Delta S+\left[N-\Delta S-\sqrt{N^2-2N\Delta S}\right]\qquad\text{for}~\Delta S\in[-\log d,0)~.
\end{align}
This, combined with (\ref{simplelookinglowerboundonrelativeentropy}), shows that for $\Delta S<0$ the above swap process (\ref{optimalswapprocess}) can never attain the lower bound on $\beta\Delta Q$ from (\ref{inequalitymaintheoremfinited}). But are there processes other than swaps (and possibly with $d_S\neq d$) that attain the bound (\ref{inequalitymaintheoremfinited}) for $\Delta S<0$?

Going through the derivation in Sections \ref{subsectDeltaQbounds} and \ref{subsectDeltaSleq0}, one actually sees that for $\Delta S<0$ the bound (\ref{inequalitymaintheoremfinited}) from Theorem \ref{maintheoremcombined} is never sharp. This is ultimately because (\ref{simplelowerboundonbetaDwithNdandDeltaQ}) is a \emph{strict} inequality for any $\Delta Q\neq0$ since the value $N(d)$ in (\ref{upperboundonvargamma}) is attained by $\var_\gamma(\gamma H)$ for \emph{at most one} value of $\gamma$, corresponding to at most one energy $E'$ in (\ref{lastequationincontinuousderviationwithequality})--(\ref{lowerboundonDwithDeltaQsquared}); this uniqueness is shown in \cite{inequalityarxivpaper}.

One might guess a better lower bound on $\beta\Delta Q$ in the case $\Delta S<0$ to be $\Delta S+M(\Delta S,d)$. This would at least be attained for the special swap process described below Eq.\ (\ref{optimalswapprocess}). On the other hand, this guess is well-defined merely for $\Delta S\geq-\log d$, since $M(\Delta S,d)$ is not defined for smaller $\Delta S$. In the quantum case, however, any value of $\Delta S\in[-2\log d,-\log d)$ is possible as well, even with $\Delta Q=0$; this follows from (\ref{upperlowerboundsonDeltaS}) and the subsequent example, if one replaces $\ket{\omega}_{S_1,S_2}$ by a general pure $\ket{\varphi}_{S_1,S_2}$ and uses the initial reservoir state $\rho_R:={\rm{tr}}_{S_2}\left[\ket{\varphi}_{S_1,S_2}\bra{\varphi}\right]$.

\section{Extended notions of Landauer processes}\label{extendednotionssectionapp}

\subsection{An integral version of Landauer's Principle}\label{integralversionsubsect}
Here we develop a modified version of Landauer's Principle, with an integral in place of the term $\beta\Delta Q$ from Theorem \ref{landauereqntheorem}. The derivation requires Lemma \ref{lemmaSEbeta} relating entropy, energy and inverse temperature, which is rigorously proven for finite dimensions in Appendix \ref{thermodynamicappendix}.

\begin{theorem}[Integral version of Landauer's Principle]\label{landauerintegralthm}Consider processes as described in Theorem \ref{landauereqntheorem}. Denote the energy of the initial reservoir state by $E_R:=\tr{H\rho_R}$, and denote by $\rho'_{R,th}$ the thermal state with energy ${\rm tr}[H\rho'_{R,th}]=E_R+\Delta Q=\tr{H\rho'_R}$. Then:
\begin{equation}
\Delta S\,+\,I(S':R')_{\rho'_{SR}}\,+\,D(\rho'_R\|\rho'_{R,th})~=~\int_{E_R}^{E_R+\Delta Q}\beta(E)\,dE~.\label{eqninintegralversion}
\end{equation}(See Appendix \ref{thermodynamicappendix} for the definition of $\beta(E)$, in particular Eq.\ (\ref{betafunctionofE}) and Lemma \ref{lemmaSEbeta}.)
\end{theorem}
\begin{proof}If $H\propto\1_R$, then necessarily $\Delta Q=0$ and we define the integral to be $0$ even though $\beta(E)$ is not well-defined in this case (see Appendix \ref{thermodynamicappendix}). The statement then follows immediately from (\ref{landauereqnshort}) since $\rho'_{R,th}=\rho_R$ as all thermal states on such a reservoir agree.

The proof in the general case starts again with the Second Law Lemma (Lemma \ref{propsecondlaw}):
\begin{align}
\Delta S+I(S':R')~
&=~\Delta~=~S(\rho'_R)-S(\rho_R)~=~\left[S(\rho'_{R,th})-S(\rho_R)\right]-\left[S(\rho'_{R,th})-S(\rho'_R)\right]~.\label{intermediateinintegralversion}
\end{align}
Denoting by $\beta'$ the inverse temperature of $\rho'_{R,th}$ (see beginning of Section \ref{subsectDeltaQbounds}), the last square brackets can, for $\beta'\in(-\infty,+\infty)$, be rewritten as (cf.\ also (\ref{pythagoreantheoremforD}) and following):
\begin{align}
\left[S(\rho'_{R,th})-S(\rho'_R)\right]~
&=~-S(\rho'_R)-\tr{\rho'_{R,th}\log\frac{e^{-\beta'H}}{\tr{e^{-\beta'H}}}}\\
&=~-S(\rho'_R)-\tr{\rho'_R\log\frac{e^{-\beta'H}}{\tr{e^{-\beta'H}}}}~=~D(\rho'_R\|\rho'_{R,th})~.\label{secondintermstepinintegralversion}
\end{align}
For $\beta'=\pm\infty$, one can explicitly verify (\ref{secondintermstepinintegralversion}), using that $\supp[\rho'_R]\subseteq\supp[\rho'_{R,th}]$ in this case.

Finally, as an entropy difference between two thermal states and due to Lemma \ref{lemmaSEbeta}, the first square brackets $[S(\rho'_{R,th})-S(\rho_R)]$ in (\ref{intermediateinintegralversion}) equals the integral on the RHS of (\ref{eqninintegralversion}). This is exactly the same step made in Eq.\ (\ref{firsttimeintegralappears}), and finally proves (\ref{eqninintegralversion}).
\end{proof}

\begin{remark}[Finiteness of the integral version]\label{finiteintversionremark}Note that statement (\ref{eqninintegralversion}) of the integral version of Landauer's Principle and in particular the term $D(\rho'_R\|\rho'_{R,th})$ is always finite due to $\supp[\rho'_R]\subseteq\supp[\rho'_{R,th}]$, which is easily verified. The finiteness of (\ref{eqninintegralversion}) is in contrast to the terms $\beta\Delta Q$ and $D(\rho'_R\|\rho_R)$ in the equality form (\ref{landauereqnshort}), which both equal $+\infty$ iff $\Delta Q\neq0$ and $\beta=\pm\infty$.

Note that the three equations (\ref{landauereqnshort}), (\ref{firsttimeintegralappears}), and (\ref{eqninintegralversion}) are consistent. However, the third one cannot be obtained directly by subtracting the first one from the second, due to ill-definedness in cases where $\beta\Delta Q=\infty$.
\end{remark}

One may consider the integral in (\ref{eqninintegralversion}) as a natural analogue of the term $\beta\Delta Q$ for a finite reservoir, especially due to the physical intuition that it must undergo temperature changes $\beta=\beta(E)$ because of its bounded heat capacity in finite dimensions (see below (\ref{heatcapacityT})). However, the quantities appearing in Theorem \ref{landauerintegralthm} are somewhat artificial: The state $\rho'_{R,th}$ may generally not appear physically in the process, just as little as any of the thermal states that give $\beta(E)$ its meaning (Appendix \ref{thermodynamicappendix}). Furthermore, the general discrete-step formulation of the process (Section \ref{setupsubsection}) dissonates with the integral expression in (\ref{eqninintegralversion}). At any rate, since (\ref{eqninintegralversion}) and (\ref{landauereqnshort}) are equivalent (through the identity (\ref{firsttimeintegralappears}), at least for $\beta\Delta Q<\infty$), all lower bounds on the heat expenditure $\Delta Q$ derived from either equality version of Landauer's Principle agree.

\subsection{Processes on several independent systems}\label{procindepsystemsappendix}
Consider $k$ systems $S_1,\ldots,S_k$, with initial states $\rho_{S_1},\ldots,\rho_{S_k}$, on each of which a separate process prepares a desired final state $\rho'_{(i)}={\rm tr}_{R_i}[U_i(\rho_{S_i}\otimes\rho_{R_i})U_i^\dagger]$ using reservoirs $R_i$ that are initially in thermal states $\rho_{R_i}$, with Hamiltonians $H_i$ and all at the same inverse temperature $\beta$. By Landauer's bound (\ref{Lineqineqthm}), the total heat dissipated in all $k$ processes satisfies
\begin{equation}\label{sumofheats}
\beta\sum_{i=1}^k\Delta Q_i~\geq~\sum_{i=1}^k\Delta S_i~,
\end{equation}
with $\Delta S_i:=S(\rho_{S_i})-S(\rho'_{(i)})$ and an obvious definition for $\Delta Q_i$ (cf.\ Eq.\ (\ref{DeltaQ})).

The question is now whether there exists a ``joint process'', acting jointly on all $k$ systems $S_1,\ldots,S_k$ and a large reservoir $R$ at inverse temperature $\beta$, such that the total heat dissipation can be less than the lower bound $\sum_i\Delta S_i$ from Eq.\ (\ref{sumofheats}). For this, we assume the systems $S_i$ to be initially uncorrelated, such that their joint state $\rho_S$ is
\begin{equation}\label{productonmanystates}
\rho_S~=~\rho_{S_1}\otimes\ldots\otimes\rho_{S_k}~.
\end{equation}
Furthermore, while the final state $\rho'_S:={\rm tr}_R[U(\rho_S\otimes\rho_R)]$ may be correlated among the different subsystems $S_i$, we demand that the reduced state on each individual system $S_i$ agrees with the desired final state $\rho'_{(i)}$ from above, i.e.
\begin{equation}
{\rm tr}_{S_1,\ldots,S_{i-1},S_{i+1},\ldots,S_k}[\rho'_S]~=~\rho'_{(i)}\qquad\forall i=1,\ldots,k\,.
\end{equation}

Then, by an easily derived chaining rule,
\begin{align}\label{chainingrule}
S(\rho'_S)~=~\sum_{i=1}^kS(\rho'_{(i)})\,-\,\sum_{i=1}^{k-1}I(S_i:S_{i+1}\ldots S_k)_{\rho'_S}~\leq~\sum_{i=1}^kS(\rho'_{(i)})~.
\end{align}
Thus, the total heat $\Delta Q$ dissipated in the joint process satisfies, by (\ref{Lineqineqthm}) and (\ref{chainingrule}),
\begin{equation}\label{lowerboundonejointprocess}
\beta\Delta Q~\geq~S(\rho_S)-S(\rho'_S)~\geq~\sum_{i=1}^kS(\rho_{S_i})\,-\,\sum_{i=1}^kS(\rho'_{(i)})~=~\sum_{i=1}^k\Delta S_i~.
\end{equation}
The heat expense $\Delta Q$ in any joint process is therefore lower bounded as $\beta\Delta Q\geq\sum_i\Delta S_i$, just like the total heat dissipated in the $k$ separate processes above, see Eq.\ (\ref{sumofheats}).

Further taking the result from Section \ref{attainingsection} that Landauer's bound can be approached arbitrarily closely with suitable reservoirs, the above inequalities show that the least amount of heat can be dissipated in the joint process with a product final state $\rho'_S=\rho'_{(1)}\otimes\ldots\otimes\rho'_{(k)}$; only in this case can the lower bound from (\ref{lowerboundonejointprocess}) be arbitrarily well achieved, just as in the $k$ separate processes.

The assumption (\ref{productonmanystates}) of initially independent systems $S_1,\ldots,S_k$ is essential for the result above. As in Section \ref{memorysection}, one is easily convinced that less heat than $\sum_i(S(\rho_{S_i})-S(\rho'_{(i)}))$ needs to be expended if the systems $S_1,\ldots,S_k$ were initially e.g.\ perfectly classically correlated.

\subsection{Processes in infinite dimensions}\label{sectioninfinitedim}
For systems and reservoirs not described by finite-dimensional Hilbert spaces, some of the quantities appearing in Landauer's Principle may not be defined or need a more careful definition. For example, both the initial and final system entropies may be infinite \cite{wehrlreview,thirringstatphys}, so that their difference $\Delta S$ is ill-defined. And the general equivalent of the thermal states (\ref{definethermalstateinsetup}) are so-called KMS states \cite{brattelirobinson,thirringstatphys}.

Some of our previous treatment, however, carries over to the case when the system $S$ and reservoir $R$ are described by separable Hilbert spaces (see also \cite{thirringstatphys}). We now presuppose this, and assume that the initial state of $S$ is a normal state $\rho_S$ with finite entropy $S(\rho_S)<\infty$. Assume further that a semi-bounded Hamiltonian $H$ is given for the reservoir $R$ such that at an inverse temperature $\beta\in(0,\infty]$ the thermal state $\rho_\beta=:\rho_R$ exists and has finite energy; the latter two conditions are, for $\beta\in(0,\infty)$, equivalent to $\tr{e^{-\beta H}}<\infty$ and $\tr{He^{-\beta H}}<\infty$, and they imply that the entropy $S(\rho_R)$ is finite as well.

Then, for any joint unitary $U$ on $SR$, the Second Law Lemma (Lemma \ref{propsecondlaw}) holds as well (all quantities remain finite, except that the case $S(\rho'_S)=S(\rho'_R)=I(S':R')=+\infty$ may occur). Going through the derivation (\ref{firstlineinproofofLeq}), one sees that $\beta\Delta Q>-\infty$ always, since $H$ is semi-bounded and $\rho_R$ had finite energy by assumption. Furthermore, $\Delta S=-\infty$ implies $\beta\Delta Q=\infty$, which one sees due to $D(\rho'_R\|\rho_R)\geq0$. Thus, the equality form (\ref{landauereqnshort}) of Landauer's Principle holds in the setup of the previous paragraph as well, when employing the usual rules of calculus with $\infty$ and when remembering that in the potentially ambiguous case $\Delta S=-\infty$ one has $I(S':R')=\beta\Delta Q=\infty$. Also, Landauer's bound $\beta\Delta Q\geq\Delta S$ (see (\ref{Lineqineqthm})) always holds.

If one considers a process just as above, but now with infinite $S(\rho_S)$ and finite $S(\rho'_S)$ (such that an infinite amount $\Delta S=\infty$ of entropy is ``erased'' from the system $S$), then one sees $S(\rho'_R)=\Delta=\beta\Delta Q=\infty$, so that Landauer's bound also holds.

Since the conditions for vanishing relative entropy and mutual information (where defined) are as in the finite-dimensional case \cite{ohyapetz}, one can check that the equality considerations from Corollary \ref{correqualityL} carry over to the above setup (with either $S(\rho_S)$ or $S(\rho'_S)$ finite) in the following way: If $\beta\Delta Q=\Delta S<\infty$, then $\Delta S=\Delta Q=0$ and Eq.\ (\ref{conditionsforequality1}) holds with an isometry $V$. But while, even for infinite-dimensional reservoirs, equality in Landauer's bound can be attained only in trivial cases, one can approach the bound arbitrarily closely for any given $\Delta S$ with processes using an infinite-dimensional reservoir (see Section \ref{attainingsection}, and also Appendix \ref{purestateerasure}).

\bigskip

In Appendix \ref{purestateerasure} we use Hamiltonians that are not merely unbounded but that have formally infinite ($+\infty$) energy levels (see also Section \ref{boundonpureness}). This is done in order to have some unpopulated levels in the initial reservoir state $\rho_R=\rho_\beta$. (Thermal states at zero temperature, $\beta=\infty$, may have such unpopulated levels as well, but they are necessarily completely mixed on their support space.) The calculations involving these Hamiltonians are formal. They can be understood as limiting processes, but exact purification as in Appendix \ref{purestateerasure} is achievable only at the limit (the approach to the limit is quantified in Section \ref{boundonpureness}). This issue is similar to the case of exactly zero temperature ($\beta=\infty$), whose physical relevance may be questioned as well.

\section{Erasure towards a pure state}\label{purestateerasure}
In Section \ref{boundonpureness} we saw that, in finite dimensions, any rank-decreasing process $\rho_S\mapsto\rho'_S$ necessarily has $\beta\Delta Q=\infty$, i.e.~requires either a zero-temperature reservoir ($\beta=\infty$) or infinite heat flow (via formally infinite Hamiltonian levels, in particular implying $\|H\|=\infty$). Thus, Landauer's bound $\beta\Delta Q\geq\Delta S$ cannot be tight for finite-dimensional processes with $\rank{\rho'_S}<\rank{\rho_S}$.

Here we show that rank-decreasing processes can come arbitrarily close to Landauer's bound by using an infinite-dimensional reservoir (with Hilbert space $\ell^2$; cf.\ also Appendix \ref{sectioninfinitedim}). To keep the notation manageable, we assume that a mixed initial qubit state $\rho_S={\rm{diag}}(s_1,s_2)$, with $s_1,s_2=1-s_1\in(0,1)$, is to be turned into a pure final state $\rho'_S={\rm{diag}}(1,0)$.

From the argument leading up to Proposition \ref{propositionlambdamin}, one can see that for such a process the initial reservoir state $\rho_R$ needs to have infinitely many unoccupied energy levels (see also last paragraph in Appendix \ref{sectioninfinitedim}),
\begin{align}\label{rhoRinfinitedim}
\rho_R~=~{\rm{diag}}(r_1,0,r_2,0,r_3,0,r_4,0,r_5,0,r_6,\ldots)\,\in\cB(\ell^2)~,
\end{align}where the $r_j$ denote the initial eigenvalues of the (potentially) non-empty levels, which we will determine below. At finite temperature, $\beta\in(0,\infty)$, this means that the energy levels of the reservoir Hamiltonian $H$ corresponding to the unoccupied levels have to be formally $+\infty$. We further choose a unitary $U$ that transforms $\rho_S\otimes\rho_R$ to the final state $\rho'_{SR}=\rho'_S\otimes\rho'_R$ with
\begin{align}
\rho'_R~&=~{\rm{diag}}(s_1r_1,0,s_2r_1,0,s_1r_2,0,s_2r_2,0,s_1r_3,0,s_2r_3,0,\ldots)\label{firstlinerhoprimeRinfinite}\\
&=:~{\rm{diag}}(r'_1,0,r'_2,0,r'_3,0,r'_4,0,r'_5,0,r'_6,0,\ldots)
\end{align}
and $\rho'_S={\rm{diag}}(1,0)$ from above. It is clear that such a unitary exists by just permuting the product basis states, since both $\rho_S\otimes\rho_R$ and $\rho'_S\otimes\rho'_R$ have the eigenvalues $s_ir_j$ for $i=1,2$, $j\in\nat$, in addition to countably many eigenvalues $0$.

We can now compute the heat flow. For this, denote the Hamiltonian energy levels by $h_j$ corresponding to the eigenvalues $r_j$ in (\ref{rhoRinfinitedim}), i.e.\ $r_j=e^{-\beta h_j}/\sum_{j'}e^{-\beta h_{j'}}$. Thus:
\begin{align}
\beta\Delta Q~&=~\sum_{j=1}^\infty\,(r'_j-r_j)\,\beta h_j~=~\sum_{j=1}^{\infty}(r_j-r'_j)\log r_j\\
&=~\sum_{j=1}^\infty(r_j\log r_j-r'_j\log r'_j)\,+\,\sum_{j=1}^{\infty}r'_j(\log r'_j-\log r_j)\\~&=~S(\rho'_R)-S(\rho_R)\,+\,D(\rho'_R\|\rho_R)~;\label{betaDeltaQforinfintite}
\end{align}
those computations are justfied since the $r_j$ will later be chosen such that $\rho_R$ and $\rho'_R$ are normalized states with finite entropies. One can easily see from (\ref{firstlinerhoprimeRinfinite}) that $S(\rho'_R)-S(\rho_R)=S(\rho_S)=S(\rho_S)-S(\rho'_S)=\Delta S$. Together with (\ref{betaDeltaQforinfintite}) this gives
\begin{align}\label{samelawforinfinitedimensions}
\beta\Delta Q~=~\Delta S\,+\,D(\rho'_R\|\rho_R)~,
\end{align}
which coincides with our finite-dimensional equality version of Landauer's Principle (Theorem \ref{landauereqntheorem}), considering that $I(S':R')=0$ here due to a pure $\rho'_S$ (see also Appendix \ref{sectioninfinitedim}). It is now easy to choose the occupation numbers $r_j$ such that $D(\rho'_R\|\rho_R)$, and thus $\Delta Q$, is finite.

To show moreover that the bound $\beta\Delta Q\geq\Delta S$ can be arbitrarily sharp, we have to find $r_j$ such that $D(\rho'_R\|\rho_R)$ in (\ref{samelawforinfinitedimensions}) is arbitrarily close to $0$. One way to do this is the following: Choose any $\varepsilon\in(0,1)$ and define
\begin{equation}
\begin{aligned}
r_1~&:=~0\,,&
r_2~&:=~\varepsilon\,,\\
r_{2k-1}~&:=~(1-\varepsilon)s_1r_k\qquad&
r_{2k}~&:=~(1-\varepsilon)s_2r_{k}\qquad&&\text{for}~k\geq2\,.
\end{aligned}
\end{equation}
One can see that the so defined $\rho_R$ is normalized with entropy $S(\rho_R)=(H_2(\varepsilon)+(1-\varepsilon)S(\rho_S))/\varepsilon$. And the relative entropy term in (\ref{samelawforinfinitedimensions}) is
\begin{align}
D(\rho'_R\|\rho_R)~&=~\sum_{k=1}^\infty s_1r_k\left(\log(s_1r_k)-\log r_{2k-1}\right)\,+\,\sum_{k=1}^\infty s_2r_k\left(\log(s_2r_k)-\log r_{2k}\right)\\
&=~\sum_{k=2}^\infty s_1r_k\left(-\log(1-\varepsilon)\right)\,+\,\sum_{k=2}^\infty s_2r_k\left(-\log(1-\varepsilon)\right)\\
&=~-\log(1-\varepsilon)~,
\end{align}
which indeed approaches $0$ as $\varepsilon\to0$.

Note however that, even in infinite dimensions, no process with $\rho_S\mapsto\rho'_S$ exists that makes the relative entropy term $D(\rho'_R\|\rho_R)$ in (\ref{samelawforinfinitedimensions}) vanish exactly (cf.\ Appendix \ref{sectioninfinitedim}): $D(\rho'_R\|\rho_R)=0$ would mean $\rho'_R=\rho_R$, which would imply that the state $\rho'_S\otimes\rho_R=\rho'_S\otimes\rho'_R=\rho'_{SR}$ (due to purity of $\rho'_S$) would have to be unitarily equivalent to $\rho_S\otimes\rho_R$. But this is possible only when $\rho_S$ was already rank-deficient.

Note finally that some state $\widetilde{\rho}'_S$ which is $\delta$-close to a given (possibly pure) state $\rho'_S$ can be reached with a finite-dimensional reservoir and with $\beta\Delta Q$ arbitrarily close to $S(\rho_S)-S(\rho'_S)$, see Section \ref{boundonpureness}.

\bibliographystyle{alpha}

\end{document}